\documentclass[11pt]{article}
\usepackage{relate}
\usepackage{times}
\usepackage{pgfplots}
\usepackage{pgfplotstable}
\pgfplotsset{compat=1.8}
\usepackage[aboveskip=8pt,belowskip=-4pt]{caption}
\usepackage{subcaption}
\usepackage{xspace} % deals with the problem of when you want spaces
\usepackage[margin=1in]{geometry}
\usepackage{textcomp}
\usepackage{balance}  % for  \balance command ON LAST PAGE  (only there!)
\usepackage{xcolor}
\usepackage{amsmath}
\usepackage{algorithm}
\usepackage{algorithmicx}
\usepackage{cite}
\usepackage{float}
\usepackage{siunitx}
\usepackage{graphicx}
\usepackage{xspace}
\usepackage{thm-restate}
\usepackage{balance}  % for  \balance command ON LAST PAGE  (only there!)
\usepackage{booktabs} % For formal tables

\usepackage{marginnote} % MAB: alternative to marginpar. Is there something better?
\usepackage{url}
\usepackage{amssymb}
\usepackage{enumitem}
\graphicspath{{./fig/}}
\usepackage{amsthm}
\usepackage{mathtools}
\usepackage[noend]{algpseudocode}

\newtheorem{thm}{Theorem}
\newtheorem{lem}{Lemma}
\newtheorem{prop}{Proposition}

\newtheorem{corollary}{Corollary}

\newtheorem{remark}{Remark}
\newtheorem{rem}{Remark}

\usepackage{comment}
\usepackage{color}
\usepackage{tikz}
\usetikzlibrary{fadings}

\usepackage{titling}
\thanksmarkseries{arabic}

\newcommand{\punt}[1]{}

\makeatletter
\makeatletter
\def\@copyrightspace{\relax}
\makeatother

\usepackage{xspace}
\newcommand{\singlecore}{single-processor cup game\xspace}
\newcommand{\SingleCore}{Single-Processor Cup Game\xspace}
\newcommand{\multicore}{multi-processor cup game\xspace}
\newcommand{\tweakedemphmulticore}{renormalized multi-processor cup game\xspace}
\newcommand{\tweakedmulticore}{(renormalized) multi-processor cup game\xspace}
\newcommand{\untweakedmulticore}{(non-renormalized) multi-processor cup game\xspace}
\newcommand{\untweakedemphmulticore}{non-renormalized multi-processor cup game\xspace}
\newcommand{\MultiCore}{Multi-Processor Cup Game\xspace}
\newcommand{\flushgame}{cup-flushing game\xspace}
\newcommand{\FlushGame}{Cup-Flushing Game\xspace}
\newcommand{\efc}{greedy algorithm\xspace}
\newcommand{\smoothefc}{smoothed greedy algorithm\xspace}
\newcommand{\multicoresmoothefc}{multi-processor smoothed greedy algorithm\xspace}
  
\newcommand{\defn}[1]       {{\textit{\textbf{\boldmath #1}}}}
\newcommand{\E}{\mathbb{E}}
 
\newcommand{\poly}{\mbox{poly}}

\newcommand{\av}{\operatorname{av}}

\date{}

\newcommand{\namedcomment}[3]{\xspace}
\newcommand{\mab}[1]{\namedcomment{mab}{red}{#1}}

\renewcommand{\epsilon}{\varepsilon}

\newcommand{\secref}[1]         {Section~\ref{sec:#1}}
\newcommand{\seclabel}[1]    {\label{sec:#1}}

\newcommand{\eqlabel}[1]    {\label{eq:#1}}

\renewcommand{\eqref}[1]          {Eq.~\ref{eq:#1}}

\date{}

\begin{document}
%!TEX root =  main.tex

\title{Achieving Optimal Backlog in Multi-Processor Cup Games}

%\title{A One-Shot Smoothing of a Classic Water-in-Cups Problem}

\author{
Michael A. Bender\thanks{Department of Computer Science, 
Stony Brook University, Stony Brook, NY 11794-2424 USA. 
\texttt{bender@cs.stonybrook.edu}.}
\and
Mart\'{\i}n Farach-Colton\thanks{Department of Computer Science, Rutgers University, Piscataway, NJ 08854 USA. \texttt{martin@farach-colton.com}}
\and 
William Kuszmaul\thanks{CSAIL, MIT, Cambridge MA 02139 USA. \texttt{kuszmaul.mit.edu}. Supported by an MIT Akamai Fellowship and a Fannie \& John Hertz Foundation Fellowship}\vspace*{0.2pt}
}

\begin{titlingpage}
\maketitle
\sloppy
\thispagestyle{empty}
%!TEX root =  main.tex

  \sloppy
  \begin{abstract}
The single- and multi- processor cup games can be used to model natural problems in areas such as processor scheduling, deamortization, and buffer management.

At the beginning of the single-processor cup game, $n$ cups sit in a row, initially empty. In each step of the game, a filler distributes $1$ unit of water among the cups, and then an emptier selects a cup and removes $1 + \epsilon$ units from that cup. The goal of the emptier is to minimize the amount of water in the fullest cup, also known as the backlog. It is known that the greedy algorithm (i.e., empty the fullest cup) achieves backlog $O(\log n)$, and that no deterministic algorithm can do better.

We show that the performance of the greedy algorithm can be greatly improved with a small amount of randomization: After any step $i$, and for any $k \ge \Omega(\log \epsilon^{-1})$, the emptier achieves backlog at most $O(k)$ with probability at least $1 -O(2^{-2^k})$. Our algorithm, which we call the \defn{smoothed greedy algorithm}, can also be interpreted as a one-shot smoothed analysis of the standard greedy algorithm.

Whereas bounds for the single-processor cup game have been known for more than fifteen years, proving nontrivial bounds on backlog for the multi-processor extension has remained open. We present a simple analysis of the greedy algorithm for the multi-processor cup game, establishing a backlog of $O(\epsilon^{-1} \log n)$, as long as $\delta$, the game's other speed-augmentation constant, is at least $1/\poly(n)$.

Turning to randomized algorithms, we encounter an unexpected phenomenon: When the number of processors $p$ is large, the backlog after each step drops to \emph{constant} with large probability. Specifically, we show that if $\delta$ and $\epsilon$ satisfy reasonable constraints, then there exists an algorithm that bounds the backlog after a given step by three or less with probability at least $1 - O(\exp(-\Omega(\epsilon^2 p))$. We further extend the guarantees of our randomized algorithm to consider larger backlogs.

When $\epsilon$ is constant, we prove that our results are asymptotically optimal, in the sense that no algorithms can achieve better bounds, up to constant factors in the backlog and in $p$. Moreover, we prove robustness results, demonstrating that our randomized algorithms continue to behave well even when placed in bad starting states.

\end{abstract}

% don't touch below this line.....

%%% Local Variables:
%%% mode: latex
%%% TeX-master: "main.tex"
%%% End:

% \newpage
% \tableofcontents
\end{titlingpage}

%!TEX root =  main.tex
 
\newpage

\section{Introduction}
% MAB: this macro is equivalent to \label{sec:intro}
\seclabel{intro}

A \defn{cup game}~\cite{DietzSl87,SleatorTa85,DietzRa91,AdlerBeFr03,LitmanMo09,BenderFeKr15,RosenblumGoTa04,ChrobakCsIm01}
is a multi-round game in which there there are $n$ cups, each
initially empty, and two players. In each round, the \defn{filler}
distributes water to cups subject to some constraints, and then the
\defn{emptier} removes water from the cups subject to other
constraints. For example, the filler may be constrained to add a total
of at most one unit of water per round, while the emptier may be
constrained to remove water from at most one cup.

The \defn{high-water mark} after a given round is the maximum amount of water in any cup. The emptier's goal is to minimize the high-water mark after each round and the filler acts as an adversary. In a cup game where the emptier uses a randomized strategy, the filler is \defn{oblivious}; that is, the filler does not know the random choices made by emptier (and thus does not know the state of the cups).

In this paper, we focus on two cup games that naturally arise in the study of processor scheduling \cite{AdlerBeFr03,LitmanMo09,DietzRa91}. 
These games (and their relaxations) have also appeared in a variety of other applications, ranging from deamortization \cite{AmirFaId95,DietzRa91,DietzSl87,AmirFr14,Mortensen03,GoodrichPa13,FischerGa15,Kopelowitz12}, to buffer management in network switches~\cite{Goldwasser10,AzarLi06,RosenblumGoTa04,Gail93}, to quality of service in real-time scheduling~\cite{BaruahCoPl96,AdlerBeFr03,LitmanMo09}.

\paragraph*{\large The \singlecore.} Consider $n$ threads, and suppose that at each time step they receive an aggregate of one new unit of work to do. The scheduler must then pick a thread to execute for that time step. The thread can then make $(1 + \epsilon)$ units of progress; here $\epsilon$ is a \defn{speed-augmentation constant}. Rephrasing the problem as a cup game, at each step the filler may distribute up to one unit of water among all of the cups, and the emptier may then select a cup and remove up to $1 + \epsilon$ units of water from it. In this context, the high-water mark is often referred to as the \defn{backlog}.

%The difficulty of this scheduling problem is that the scheduler must discretize a computational resource among $n$ threads with continuous demands\todo{sounds jargony}.

The \singlecore and its variants have been studied extensively \cite{AdlerBeFr03,DietzRa91,DietzSl87,BaruahCoPl96,Liu69}. Notably, if the emptier follows the greedy algorithm, and always removes water from the fullest cup, then the backlog will never exceed $O(\log n)$ \cite{AdlerBeFr03}. In fact, this remains true even when the speed-augmentation constant $\epsilon$ is zero. It is further known that no matter the value of $\epsilon$, no deterministic algorithm can do better, even by an additive constant factor (see the variant of the lower bound presented in \cite{BenderFeKr15} and \cite{DietzRa91}).\footnote{The fact that the bounds are unaffected by the choice of $\epsilon$ has led past authors to implicitly either take $\epsilon$ to be zero \cite{AdlerBeFr03,Liu69,BaruahCoPl96} or to be very large \cite{DietzSl87, DietzRa91, BenderFeKr15}. Combining the lower bound for the latter case with the upper bound for the former, one obtains the result for all values of $\epsilon$ with no additional work.}

%% \paragraph*{The \smoothefc}

\smallskip

{\bf Our single-processor results.} We show that the performance of the greedy emptying algorithm for the \singlecore can be greatly improved with a small amount of randomization: After any step $i$, and for any $k \ge \Omega(\log \epsilon^{-1})$, the emptier achieves backlog at most $O(k)$ with probability at least $1 -O(2^{-2^k})$. An important consequence is that with high probability in $n$ (i.e., probability $1 - 1/\poly(n)$), the backlog will not exceed $O(\log \log n)$, as long as $\epsilon \ge O(1/\log n)$. Moreover, we show that these bounds are optimal, meaning that no algorithm can obtain the same probabilistic bounds for asymptotically smaller backlogs $k$.

The algorithm, which we call \emph{the \smoothefc}, works as follows: The emptier begins the game by artificially inserting into each cup $j$ a random quantity $r_j \in [0, 1]$. Then at the end of each step, the emptier follows the deterministic \efc, except with a small but critical modification: If the fullest cup contains less than one unit of water, then the emptier does not remove any water at all.

This simple algorithm gives an explanation, in the form of a smoothed analysis, for why real-world systems should be expected to exhibit good behavior when following greedy strategies. Indeed, by randomly perturbing the start-state of the cups, a deterministic algorithm can suddenly obtain the guarantees of the (optimal) randomized algorithm. In this sense, our randomized algorithm can be interpreted as a one-shot smoothed analysis of the standard greedy algorithm.

The small amount of resource augmentation used by the \smoothefc makes it robust to the setting in which cups begin in a bad initial starting state. In particular, we show that if $b$ units of water are maliciously placed into cups at the beginning of the game, then for steps $i > \frac{b}{\varepsilon}$, the $b$ units of water have no affect on the guarantees given by the algorithm.

\paragraph*{\large The \multicore.} The multi-processor version of the same scheduling question has proven to be much harder. In each step of the \defn{\multicore}, the filler distributes some amount of water proportional to $p$, the number of processors, among the cups (i.e., threads), and the emptier picks $p$ cups and removes a unit of water from each.  If the filler is unrestricted in their placement of the water, then they can ensure an arbitrarily large backlog by placing all of the water in the same cup at each step; the emptier will then only be able to remove a single unit from that cup. A natural restriction, therefore, is to bound the amount of water placed in each cup at each step by some quantity $1 - \delta$, where $\delta$ is, like $\epsilon$, a speed-augmentation constant

The difficulty of the multi-processor version of the question stems from the fact that the emptier can remove at most one unit of water from each cup. In the scheduling problem, this corresponds with the fact that distinct processors cannot make progress on the same thread simultaneously, even if that thread has a vastly greater backlog of work than its peers. As noted by Liu \cite{Liu69}, and subsequently reiterated by later authors \cite{LitmanMo09,BaruahCoPl96}, this adds a ``surprising amount of difficulty'' to the scheduling problem.

Litman and Moran-Schein \cite{LitmanMo09} showed that if the emptier has speed-augmentation constant $\delta$ and is permitted to be semi-clairvoyant, with knowledge of the schedule for when the next $1 + \delta^{-1}$ units of water will arrive in each cup, then the emptier can follow a simple deadline-based algorithm to achieve backlog $O(1 + \delta^{-1})$. Providing provable guarantees for a fully \emph{non-clairvoyant} emptier has been a long-standing open problem~\cite{LitmanMo09,AdlerBeFr03,BaruahCoPl96}. It has been unknown, for example, whether a sufficiently clever filler might be able to achieve a backlog as a function of $p$.

%% \paragraph*{Provable guarantees for the \multicore}
\textbf{Our multi-processor results.}
We present the first provable results for the \multicore without the use of clairvoyance. Specifically, we consider the game in which at each step, the filler places $(1 - \epsilon)p$ units of water among the cups, putting no more than $1 - \delta$ units of water into any particular cup; the emptier then removes up to one unit of water from each of $p$ cups of their choice. We prove that the greedy Empty-Fullest-Cups algorithm achieves a backlog of $O\left(\frac{1}{\epsilon} \log n\right)$ for any $0 < \epsilon < 1$ and $\delta \ge \frac{1}{\poly(n)}$. When $\epsilon$ is constant, we show that this is provably optimal up to constant factors.

Turning to randomized algorithms, we encounter an unexpected phenomenon: When the number of processors $p$ is sufficiently large, the backlog after each step drops to \emph{constant} with large probability. Specifically, we show that, if $1/2 \ge \varepsilon \ge p^{-1/3}$ and $\delta < 1$ is sufficiently large in $\Omega(e^{-O(\varepsilon^2p)})$, then there exists an algorithm that guarantees a backlog of three or less after each step with probability at least $1 - O\left(e^{-\Omega(\varepsilon^2p)}\right)$.

We further extend the guarantees of our randomized algorithm to consider arbitrary backlogs. Namely, if we add the restriction that $\delta \ge \frac{1}{\poly(p)}$, then our algorithm achieves backlog $O(k / \epsilon)$ after each step with probability at least $1 - O\left(e^{-e^{k}}\right)$. Additionally, in the case where $\epsilon$ is constant, we prove lower bounds establishing that our probabilistic guarantees give optimal bounds up to constant factors in the backlog $k$ and in $p$.

As in the single-processor case, we show that the \smoothefc is robust to the setting in which cups begin in a bad initial starting state. In particular, suppose $b$ units of water are maliciously placed into cups at the beginning of the game. Then the above guarantees for the \multicoresmoothefc continue to hold after any step $i \ge 2b / \delta$, except with an additional failure probability of $O\left(e^{-\Omega(i^{1/6})}\right)$.

\paragraph*{\large A relaxed game for deamortization: the \flushgame.} A common application of cup games is to the deamortization of data structures \cite{AmirFaId95,DietzRa91,DietzSl87,AmirFr14,Mortensen03,GoodrichPa13,FischerGa15,Kopelowitz12}. In this setting, a relaxation of the \singlecore, also known as the \defn{\flushgame}, is often used. In the \flushgame, the emptier is permitted to remove \emph{all} of the water from a cup at each step, rather than just a single unit. This is sometimes referred to as \defn{flushing} a cup.

%% In each step of the cup-flushing game, the filler distributes one unit of water among the cups, and then the emptier selects a single cup and removes \emph{all} of the water from it.  A greedy emptier---who always empties a fullest cup---can keep the backlog to $O(\log n)$ \cite{DietzSl87}. Using a more sophisticated randomized algorithm, the emptier can maintain the backlog at $O(\log\log n)$ with high probability after any given round of the game \cite{DietzRa91}. \bill{This is no good. We haven't talked about why the deamortization is so cool at all!}

As in the \singlecore, a greedy emptier---who always flushes a fullest cup---can keep the backlog to $O(\log n)$ \cite{DietzSl87}. Using a more sophisticated randomized algorithm, the emptier can maintain the backlog at $O(\log\log n)$ with high probability in $n$ after any given round of the game \cite{DietzRa91}. 

\textbf{Relationship with our results.}  The complicatedness of the randomized algorithm prompted the authors to pose finding a simpler algorithm as an open problem \cite{DietzRa91}. Since the \flushgame is a relaxation of the \singlecore, our randomized algorithm for the latter provides a resolution to this open problem. Our algorithm additionally improves on the prior work by offering probability bounds against any backlog $k$, rather than applying only to $k = \log \log n$.

%describe their randomized algorithm as being ``somewhat complex'', and

In should be noted that, in general, the \flushgame can offer substantially more power to the emptier than do the scheduling games. Consider, for example, the ``multi-processor'' version of the \flushgame, in which the filler places $p$ units of water among cups and then the emptier flushes the water from $p$ cups. A solution to the ``multi-processor'' version follows trivially from the solution to the ``single-processor'' version, since the emptier can simply simulate the water poured by the filler as appearing in chunks of size one over the course of $p$ steps, and then perform the corresponding $p$ cup flushes (some of which may be to the same cup as each other). In contrast, the standard techniques for analyzing the \singlecore do not seem to generalize to the multi-processor case, and indeed the latter has remained open until now.  

\paragraph*{\large Other related work.} Adler et al. \cite{AdlerBeFr03}, who proved the $O(\log n)$ bound on the backlog in the \singlecore, also considered a relaxation of the \multicore in which the $p$ units of water removed by the emptier need not be from distinct cups, establishing a tight bound of $\Theta(\log n)$ backlog for this problem (without the use of clairvoyance or speed augmentation). As noted by \cite{BaruahCoPl96}, however, this relaxed version of the problem can be reduced to the single-processor version of the same question, and is thus primarily interesting from the perspective of lower bounds rather than upper bounds. Litman and Moran-Schein \cite{LitmanMo09} proved a formal separation between the relaxed version and the non-relaxed versions, establishing that a deadline-based algorithm (which is completely clairvoyant) does provably worse in the multi-processor setting when it is not permitted to run the same thread on multiple processors concurrently.

In the context of packet-switching, Bar-Noy et al. \cite{Bar-NoyFrLa02} considered a variant of the \singlecore in which the filler is unconstrained as to how much water can be placed in the cups at each step, but must always place water in integer amounts. Rather than providing absolute bounds on backlog, which would be impossible in this scenario, they show that the \efc achieves an $O(\log n)$ competitive ratio with the optimal offline emptying algorithm. Moreover, they show that no online algorithm, including randomized algorithms, can do better. Several of the same results were also discovered concurrently by Fleischer and Koga \cite{FleischerKo04}. Subsequent work has extended this to consider competitive ratios against other weaker adversaries \cite{DamaschkeZh05}.

Extensive work has also focused on variants of the cup games in which the rates of arrival are fixed \cite{BaruahCoPl96,GkasieniecKl17,BaruahGe95,LitmanMo11,LitmanMo05,MoirRa99,BarNi02}. That is, each cup $j$ receives the same amount $c_j$ from the emptier at each step $i$. Even in this simplified setting, constructing fast algorithms is highly non-trivial. Notably, by exploiting certain results from network flow theory, Baruah et al. \cite{BaruahCoPl96} present a polynomial-time algorithm (per step of the game) that the emptier can follow in order to guarantee a maximum backlog of $O(1)$ in the \singlecore with fixed rates.

Additionally, a number of papers have explored variants of the \multicore in which the cups form the nodes of a graph, and the constraints on the emptier are a function of the graph structure \cite{BodlaenderHuKu12,BenderFeKr15,BodlaenderHuWo11,ChrobakCsIm01}. This is important, for example, in the study of multiprocessor scheduling with conflicts between tasks \cite{BodlaenderHuWo11,ChrobakCsIm01}, and in the study of sensor radio networks \cite{BenderFeKr15}.

Recently, cup emptying games have also found applications to modeling memory-access heuristics in databases \cite{BenderCrCo18}. Here, whenever the emptier removes the water from a cup, the water is then redistributed according to some fixed probability distribution. The emptier's goal is ensure that, on average, the amount of water in the cup being emptied is large.

\section{Preliminaries}
\label{sec:prelim}

We describe formally the conventions that we will follow in discussing the cup games.

\smallskip

\noindent\textbf{The \singlecore.} A sequence of $n$ cups sit in a
line, initially empty. At each step in the game, a \defn{filler}
distributes (up to) $(1 - \varepsilon)$ units of water across the
cups. Then an \defn{emptier} selects one cup and removes (up to) one
unit of water from the cup. The goal of the emptier is to minimize the amount of water in the fullest cup after each step, also known as the \defn{backlog}.\footnote{Note that one could just as reasonably follow the convention that the filler places $1$ unit at each step, and the emptier removes $1 + \epsilon$ units (as was the case in Section \ref{sec:intro}). For $\epsilon \le 1/2$, the distinction between the conventions is aesthetic only, up to constant-factor changes in $\epsilon$.}

The following notation will be useful for analyzing the
\singlecore. For each $i \in \mathbb{N}$, and for each cup-number $j
\in [n]$, define $c_j(i)$ to be the amount poured into cup $j$ during
the $i$-th step. The quantities $c_j(i)$ are predetermined by the
filler at the start of the game (i.e., the filler is an oblivious
adversary), and are constrained to satisfy $\sum_{j \in [n]} c_j(i)
\le 1 - \varepsilon$. Let $f_j(i)$ denote the amount of water in cup
$j$ after step $i$. The backlog of the system after step $i$ can be
expressed as $\max_j f_j(i)$, the amount of water in the fullest cup.

\smallskip

\noindent\textbf{The \multicore.}  As in the
\singlecore, a sequence of $n$ cups sit in a line, 
initially empty. 
\mab{Why is there a sequence? Is the fact that there's a line important? 
I don't think so, so I believe that we should remove these words.}
At each step in the game, the filler distributes (up
to) $(1 - \varepsilon)p$ units of water among the cups, with no cup
receiving more than $1 - \delta$ units of water at a time. The emptier
then removes (up to) one unit of water from (up to) $p$ distinct
cups. As before, the goal of the emptier is to minimize the amount of
water in the fullest cup at any given step, also known as the
\defn{backlog}. As in the \singlecore, the filler is an oblivious
adversary.

We define the quantities $f_j(i)$ and $c_j(i)$ as in the \singlecore. The restrictions on the emptier translate to the
constraint that, for a given step $i$, we have $c_j(i) \le 1 - \delta$ for
each cup $j$, and $\sum_j c_j(i) \le (1 - \varepsilon) p$.

\smallskip

\noindent\textbf{The \tweakedemphmulticore.}  It is sometimes useful
to use an (aesthetic) variation on the \multicore, which we refer to
as the \defn{\tweakedemphmulticore}, in which the filler distributes
$(1 - \varepsilon)p$ units of water among the cups, placing no more
than one unit in any cup (rather than $1 - \delta$); and the emptier
then selects up to $p + 1$ (rather than $p$) distinct cups and removes up to $1 + 2\delta$ from each (rather than $1$).

For convenience of algorithm analysis, we use the
\tweakedemphmulticore when designing randomized algorithms for the
\multicore, and we use the \untweakedmulticore when analyzing
deterministic algorithms. All of our results are portable between two
versions of the game, with changes in $\epsilon$ and $\delta$ by
constant factors. See Appendix \ref{apptweaked} for more details.

\smallskip

\noindent\textbf{Bounding the number of random bits.} Many of our
algorithms assume the ability to select random real thresholds in the
range $[0, 1]$. In Appendix \ref{apprandombits} we describe how to use
low-precision thresholds in order to achieve the same guarantees,
thereby reducing the number of random bits to $\Theta(\log n)$ per
threshold.
\mab{This is a result, but it's only mentioned in \secref{prelim}. Shouldn't we at least mention it earlier when we discuss our results?}

%!TEX root =  main.tex

\section{Technical Overview}
% MAB: section label should be the file name with a sec: in front
\label{sec:technical}

\paragraph*{A randomized algorithm for the \singlecore.} In \secref{singlecore}, we present a randomized algorithm for the \singlecore with the following guarantee: After any step $i$, and for any $k \ge \Omega(\log \epsilon^{-1})$, the emptier achieves backlog at most $O(k)$ with probability, i.e., probability at least $1 -O(2^{-2^k})$. When $\epsilon \ge O(1/\log n)$, it follows that with high probability in $n$, the backlog will not exceed $O(\log \log n)$.

The algorithm, which we call \defn{the \smoothefc}, works as follows: The emptier begins the game by artificially inserting into each cup $j$ a random quantity $r_j \in [0, 1]$. Then at the end of each step, the emptier follows the deterministic \efc of emptying from the fullest cup, except with a small but critical modification: If the fullest cup contains less than one unit of water, then the emptier does not remove any water at all. This modification ensures that the actions of the emptier do not affect the amount of water modulo one in each cup, a property which will play a critical role in the analysis of the algorithm.

To analyze the \smoothefc, we begin by revisiting the deterministic \efc in a new setting in which the number of cups can change dynamically: arbitrarily many new empty cups may be added by the filler at the beginning of each step, and cups are removed at the end of a step if they are empty. We call this the \defn{dynamic \singlecore}. We prove that if $n_i$ is the number of cups present at the end of a step $i$, then the backlog does not exceed $O(\log n_i)$.

The analysis of the \smoothefc reduces the game on $n$ cups to a dynamic game on far fewer cups. Call a cup \defn{active} if it contains one or more units of water. The \smoothefc can be seen as playing a virtual dynamic game on the active cups in which the virtual fill of each cup is the true fill minus one. In particular, cups enter and and leave the dynamic game based on whether they have positive virtual fill, and the emptier always removes water from the fullest cup. To bound the backlog by $O(k)$, it therefore suffices to bound the number of active cups by $2^k$.

The difficulty in bounding the number of active cups directly is that there may be a series of time steps during which the emptier removes water exclusively from cups containing more than two units of water, but during which the filler manages to increase the number of active cups significantly.

We instead bound a different quantity, which we call the \defn{integer fill}. The integer fill after step $i$ is $\sum_{j = 1}^n \lfloor f_j(i) \rfloor$, and can be thought of as the number of \defn{integer thresholds} crossed within the cups. This gives an upper bound for the number of active cups, since each active cup contributes at least 1 to the integer fill. On the other hand, the integer fill also interacts nicely with the emptier, because whenever the emptier removes water from a cup, the emptier decreases the integer fill by exactly $1$. 

This property ensures that, in order for the integer fill to exceed $2^k$ after a step $i$, the following must be true: There must be some $l > 0$ so that in the $l$ steps leading up to $i$, at least $2^k + l$ integer thresholds were crossed by the filler. We call any $l$ satisfying the latter property a \defn{backlog witness}. 

To bound the probability that  there exists a backlog witness $l > 0$, we begin by considering the simpler problem of bounding the number of integer thresholds crossed in a single step. Recall that in the \smoothefc, the emptier always removes water in integer quantities. Thus, the amount of water modulo 1 in a cup $j$ after step $i$ is 
\begin{equation}
r_j + \sum_{t = 1}^i c_j(t) \pmod 1,
\eqlabel{eqmod1isgood}
\end{equation}
regardless of the actions of the emptier.

Although the random quantity $r_j$ of water is inserted into the cup $j$ only once at the start of the game, it continues to randomize the amount of water modulo 1 at each step $i$. Remarkably, this makes the probability that cup $j$ crosses an integer threshold in step $i$ exactly $c_j(i)$. Moreover when one cup crosses thresholds is independent of when other cups do.

Naturally, the number of integer thresholds crossed in one step is not independent of previous steps; nonetheless, we can analyze $l$ consecutive steps together and express the number of integer thresholds crossed as a sum of independent indicator random variables. 
If $u_j$ water is placed in cup $j$ during the $l$ steps, then the cup simply contributes $\lfloor u_j \rfloor$ indicator random variables that are deterministically set to one, and one indicator random variable whose value is $1$ with probability equal to the fractional part of $u_j$.

The number of integer thresholds crossed by the filler in any sequence of $l$ consecutive steps is therefore a sum of independent indicator random variables with mean $(1 - \epsilon)l$. Applying a Chernoff bound and then summing over all $l > 0$, we deduce that the probability of there existing a backlog witness is $O(2^{-2^{\Omega(k)}})$.

\paragraph*{Analyzing the deterministic \efc in the \multicore.}
Past analyses of the \singlecore---which have relied on showing that a collection of $n$ invariants remain simultaneously true at the end of each step \cite{AdlerBeFr03,DietzSl87}---have so far failed to generalize to the \multicore. This difficulty has led authors to consider weaker variants of the game in which either the emptier is given clairvoyant abilities \cite{LitmanMo09}, or is permitted to remove as much water as it wants from individual cups \cite{AdlerBeFr03}. The sparsity of results in the multi-processor setting is not for lack of trying; even the earliest papers on the subject \cite{Liu69, BaruahCoPl96} explicitly acknowledge the difficulty of the problem.

In \secref{multicoredeterministic}, we prove bounds for the \multicore using the parallel \efc.
In particular, when $\epsilon \ge \Omega(1)$ and $\delta \ge \frac{1}{\poly(n)}$, we achieve a provably optimal bound of $O(\log n)$ on the backlog, matching the performance in the \singlecore. More generally, for any $\epsilon > 0$ and $\delta \ge \frac{1}{\poly(n)}$, our algorithm achieves backlog $O(\epsilon^{-1} \log n)$. In this discussion we will assume that $\epsilon \ge 1/n$, so that $\delta \ge \epsilon^{-1} / \poly(n)$, although this requirement is removed in the full proof.

Rather than building on the invariants-based approach taken by past authors \cite{AdlerBeFr03,DietzSl87}, we show that a surprisingly simple potential function can be used to bound the backlog. Define,
$$
\phi(i) = \sum_{i = 1}^n (1 + \epsilon)^{f_j(i)}.
$$

We show that $\phi(i) \le \poly(n)$ for all steps $i$. It follows that no $f_j(i)$ can ever exceed $O\left(\frac{1}{\epsilon} \log n\right)$.

Fix $n^c$ to be a sufficiently large polynomial. To bound $\phi(i)$, there are three cases:

\noindent \emph{Case 1: The emptier removes a full unit of water from
  each of $p$ cups during step $i$:} In this case, the water placed by
  the filler in step $i$ can be matched with the water removed by the
  emptier in a way so that each small amount of water poured is
  matched with at least $(1 + \epsilon)$ times as much water that is
  removed from a height no more than one smaller. Since
  $$\frac{d}{dx}\left((1 + \epsilon)^{x}\right) = (1 + \epsilon) \cdot
  \frac{d}{dx}\left((1 + \epsilon)^{x - 1}\right),$$ it follows that
  the net change in potential during the step is non-positive.

\noindent \emph{Case 2: For all cups $j$, $f_j(i) \le \log_{1 + \epsilon} n^{c - 1}$.} In this case, we trivially
have that $\phi(i) \le n^c$, since each cup contributes at most $n^{c - 1}$ to $\phi(i)$.

\noindent \emph{Case 3: After step $i$, at least one cup contains
  $\log_{1 + \epsilon} n^{c - 1}$ or more water, but during step $i$
  there are never more than $p - 1$ cups containing one or more units
  of water.} This case is the most subtle. The key insight is that for
  each of the cups $c_1, \ldots, c_t$ that contain more than one unit
  of water at the end of step $i$, all of the water that was placed
  into the cup during the step was subsequently removed. Moreover, at
  least $\delta$ additional water is successfully removed from each
  of the cups $c_1, \ldots, c_t$, bringing their total contribution to
  the potential down by a factor of $(1 + \epsilon)^{\delta} = 1 +
  \Theta(\epsilon \cdot \delta)$. Since at least one of the cups $c_1,
  \ldots, c_t$ contains $\log_{1 + \epsilon} n^{c - 1}$ units of
  water, this drop in potential is at least $n^{c - 1} \cdot
  \Theta(\epsilon \cdot \delta)$. Using that $\delta \ge \epsilon^{-1} /
  \poly(n)$, if $c$ large enough, the drop in potential is at least
  $2n$. On the other hand, the total contribution to $\phi(i)$ by all
  other cups at the end of the step is at most $(1 + \epsilon) \cdot n
  \le 2n$. Thus, we get that $\phi(i) \le \phi(i - 1)$.

\vspace{.3 cm}

The simplicity of the potential-function argument makes \secref{multicoredeterministic} among the shortest in the paper. Nonetheless, due to the history of the problem, we consider it to be among our main results. It is also worth remarking that the analysis above can be used to reprove a number of classical results for cup games, such as the fact that the \flushgame achieves backlog $O(\log n)$. In fact, we can prove a stronger statement, which is that the backlog remains $O(\log n)$ even if the emptier's approach is relaxed to always empty out of a cup that is within an additive constant of being fullest.

\paragraph*{A randomized algorithm for the \multicore.} 
In \secref{multicoreconstantbacklog}, we present a randomized
algorithm for the \tweakedmulticore. Specifically, we show that, if
$1/2 \ge \varepsilon \ge p^{-1/3}$ and $\delta < 1$ is sufficiently
large in $\Omega(e^{-O(\varepsilon^2p)})$, then there exists an
algorithm that guarantees the backlog is three or less after each step
with probability at least $1 -
O\left(e^{-\Omega(\varepsilon^2p)}\right)$. In particular, when
$\varepsilon$ is constant and the number of processors $p$ is
relatively large, this brings the probability of super-constant
backlog down to exponentially small in $p$. In contrast, in the
single-processor case, no bounds on constant backlog can be achieved
with better than constant probability.

Because this is the most technically difficult result in the paper, we give a somewhat higher level overview. As in the \singlecore, we once again apply our technique of using randomized threshold crossings as an accounting scheme to track the emptier's progress. However, the interactions between multiple cups being emptied simultaneously makes both the instantiation of the thresholds and the analysis of the algorithm much more delicate. 

The random thresholds are defined within the cups so that whenever $c_j(i)$ new units of water are placed in a cup $j$, at most one threshold is crossed, and the probability of a threshold being crossed is at most $c_j(i)$. Moreover, whenever a threshold is crossed, the cup $j$ is guaranteed to have enough water (i.e., $1 + \delta$ units) that the emptier can make close to their maximum allowable withdrawal (of $1 + 2\delta$ units) from the cup. Whenever $p$ or fewer thresholds are crossed during a step, the emptier always has the option of removing water from the cups in which thresholds were crossed. When this happens, we regard the emptier as having been more successful during the step than was the filler.

On the other hand, if $k > p$ thresholds are crossed during a step $i$, then we define the number of \defn{surplus thresholds} crossed as $T_i = k - p$. When the number of processors $p$ is large, the probability of there being \emph{any} surplus threshold crossings for a given step $i$ becomes small. In particular, since the total amount of water $\sum_j c_j(i)$ placed by the filler at each step $i$ is at most $(1 - \epsilon) p$, and since threshold crossings are independent between cups, we can apply a Chernoff bound to bound $\Pr[T_i > 0] \le e^{-\Omega(\epsilon^2 p)}$.

The key to analyzing the algorithm after step $i$ is to show that with large probability, all $l > 0$ satisfy the property that $\frac{\delta}{2} \cdot l$ or fewer surplus threshold crossings occurred in the $l$ steps leading up to step $i$. In particular, the emptier's algorithm is designed to make progress roughly $\delta$ at each step towards undoing any damage caused by past surplus threshold crossings; if all $l > 0$ satisfy the aforementioned property, then we are guaranteed after step $i$ that the emptier has successfully handled all surplus thresholds from the past.

The technical challenge becomes to prove a concentration bound on the number $\sum_{m = i - l + 1}^i T_m$ of surplus threshold crossings within a sequence of $l$ steps. This is made difficult by the fact that there exists a complicated chain of dependencies between different $T_m$'s. In particular, the fact that a given $T_m$ is greater than zero cannot be attributed to any single cup. Rather, it requires that at least $p + 1$ cups all have thresholds crossed simultaneously in step $m$. Which $p + 1$ cups these are can then greatly impact for other steps $m'$ the probability that $T_{m'}$ is greater than zero.

To prove a concentration bound we use McDiarmid's inequality (i.e., Azuma's inequality applied to Doob Martingales) \cite{McDiarmid89, AlonSp04}. McDiarmid's inequality says that if a random variable $Y$ can be expressed as a function $F(X_1, \ldots, X_t)$ of independent random variables $X_1, \ldots X_t$, and if each $X_j$'s value can affect $F(X_1, \ldots, X_t)$ by at most some bounded amount $c$, then a Chernoff-style concentration bound holds for $Y$:
$$\Pr[Y \ge \E[Y] + R] \le e^{-2R^2 / (c^2 t)}.$$

In order to make the algorithm more conducive to an application of McDiarmid's inequality, a natural way to design the thresholds in our algorithm would be to have the thresholds be constantly changing within each cup $j$. In particular, one could select a random threshold $r_j^s \in [s, s + 1]$ for each integer $s$, and say that a threshold is crossed within cup $j$ whenever the total amount of water ever poured into the cup passes a threshold $r_j^s$. Unfortunately, such a design would allow for two thresholds in a single cup to potentially be crossed in the same step, making it no longer be the case that the first $p$ threshold-crossings in each step are easily handled by the emptier.

To avoid this problem, we use a more nuanced layout of thresholds, in which gaps are placed between thresholds that differ in their values modulo 1. In order so that these gaps are not too frequent (which would then make it so that the threshold crossings no longer closely tracked the amount of water in the cups), thresholds are then placed in collections of size $1/\delta$, with each collection of thresholds using the same random offset modulo 1, and with consecutive collections being separated by a gap of size one. We define a \defn{threshold collection} to be a set of thresholds that use the same random offset modulo 1, and we define the \defn{random value} of the collection to be that random offset. We define the \defn{hitting capacity} of a threshold collection in cup $j$ during steps $i - l + 1, \ldots, i$ to be the expected number of thresholds (depending on the random value of the collection) that the water poured during those steps will have crossed in the collection.

With thresholds designed in this manner, a first attempt at applying McDiarmid's inequality might go as follows: Notice that $Y = \sum_{m = i - l + 1}^i T_m$ can be written as a function $F$ of the random values of each of the threshold collections with non-zero hitting capacities. As a loose bound, there are at most $n + lp$ such threshold collections, since in order for a cup $j$ to have $h > 1$ such threshold collections, at least $(h - 1)$ units of water must have been poured into the cup during the steps $i - l + 1, \ldots, i$ (that way the water could cross the threshold-less gap between the collections). Moreover, the random value of each threshold collection can affect $Y$ by at most $\delta^{-1}$, since the random value controls in total when at most $\delta^{-1}$ different threshold crossings occur. Thus we can apply McDiarmid's inequality with $c = \delta^{-1}$ and $t = n + lp$.

This na\"\i{}ve approach only gives a useful bound when $lp \ge n$. To handle smaller $l$, a more sophisticated argument is necessary.  We express $Y$ as a function of random variables $X_1, \ldots, X_{pl + 1}$ where each $X_s$ corresponds with the random values of some set of threshold collections rather than with just a single random value. That is, we partition the threshold collections into a partition $Q_1, \ldots, Q_{pl + 1}$, and then define for each component $Q_s$ a random variable $X_s$ that reveals the random values of the threshold collections in the component $Q_s$. The key is to design the partition $Q_1, \ldots, Q_{pl + 1}$ so McDiarmid's inequality can subsequently be applied. This is accomplished by assigning threshold collections that are likely to have substantial affect on $Y$ to their own components, and then grouping together large sets of threshold collections that are not expected to have substantial affect on $Y$ into other components. We show that with a large probability, no individual $X_s$ has substantial affect on the outcome of $Y$, allowing us to obtain the desired concentration bound via McDiarmid's inequality.

\paragraph*{The \multicoresmoothefc.} In the algorithm described
above, whenever there are multiple cups containing super-constant
amounts of water, the emptier may select which of them to empty out of
arbitrarily. We define the \defn{\multicoresmoothefc} to be the
strategy in which, when deciding between cups of equal priority, the
emptier always chooses the fullest cups.

Adding a new condition that $\delta \ge \frac{1}{\poly(p)}$, we extend
the analysis of the randomized algorithm for the \multicore to
consider arbitrary backlogs. Specifically, in Section
\ref{sec:multicoresmoothefc}, we show that after any step $i$ and for
any $k > 1$, the \multicoresmoothefc achieves backlog $O(k /
\epsilon)$ with probability at least $1 - O\left(e^{-e^{k}}\right)$.

The analysis combines ideas from each of the previous sections. We
consider a dynamic version of the \multicore being played on the cups
containing $4 + \delta$ or more units of water. The number of cups
involved in the game at the end of a step $i$ is upper bounded by the
number of surplus threshold crossings that have been crossed in the
past without being subsequently accounted for by the
emptier. Concentration bounds for the number of cups in the dynamic
game can thus be proven using extensions of the ideas from
\secref{multicoreconstantbacklog}. Combining this with an analysis of
the deterministic \efc for the dynamic \multicore, we obtain the
desired bound on backlog.

\paragraph*{Matching lower bounds. }In Section
\ref{sec:lowerbounds}, we show that when $\epsilon$ is constant, the
bounds presented above for the \smoothefc, the multi-processor \efc,
and the \multicoresmoothefc are asymptotically optimal, in the sense
that no algorithms can achieve better bounds, up to constant factors
in the backlog and in $p$.

The strategy that the filler follows is a variant on the
strategy used by past authors in the \singlecore
\cite{BenderFeKr15, DietzRa91}. In particular, in the $p$-processor
setting, the filler uses the following algorithm: First
distribute $(1 -\epsilon)p$ units of water among some number $m$
of cups. Next attempt to guess which $p$ cups the emptier removes water
from, and then distribute $(1 - \epsilon)p$ units of additional water
among the remaining $m - p$ cups. Again attempt to guess which cups
have water removed from them, and continue like this until there are
$p$ or fewer cups remaining. In the event that the filler has made all
their guesses correctly, the fullest cup will have
\begin{align*}
  & p \cdot (1 - \epsilon) \cdot \left(\frac{1}{m} + \frac{1}{m - p} + \frac{1}{m - 2p} + \cdots \right) \\
  & = \Theta\left(\frac{1}{m/p} + \frac{1}{m/p - 1} + \frac{1}{m/p - 2} + \cdots \right)  = \Theta\left(\log \frac{m}{p}\right).
\end{align*}
Each of our lower bounds are proven simply by analyzing the success
probabilities of variants of this strategy for various parameters $m$
and $p$.

\paragraph*{Recovering from bad starting states.} In 
\secref{recovery}, we revisit both the \smoothefc and the \multicoresmoothefc in the situation where the initial state of the cups is no longer
empty. In particular, suppose that $b$ units of water have already
been dispersed among the cups arbitrarily before the game begins. We
show that in both games the system can quickly recover from such a
starting state.

For the \singlecore, for $i > \frac{b}{\varepsilon}$, the \smoothefc maintains the same guarantees as if the initial states of the cups
were empty. The key insight is that there must be some step $i' \le i$
after which the integer fill is zero; using this fact, the analysis of the \smoothefc at step $i$ can then be applied without modification. To
prove the existence of such an $i'$, notice that for each step, either
the integer fill is zero at the end of the step, or the emptier must
have successfully removed a full unit of water during the step. In the
latter case, the total amount of water in the system will have decreased
by at least $\delta$. This implies that the integer fill cannot be
non-zero at the end of each of the steps $1, \ldots, i$ for any $i >
\frac{b}{\varepsilon}$, as desired.

For the \multicore, the guarantees previously proven for the
\multicoresmoothefc continue to hold after a step $i \ge 2b / \delta$
with an additional failure probability of
$O\left(e^{-\Omega(i^{1/6})}\right)$. This corresponds with the
probability that the emptier fails to achieve a state in which every
cup is nearly empty during any of the steps $i' \le i$.

%!TEX root =  main.tex

\section{A Smoothed Algorithm for the \SingleCore}

% MAB: seclabel actually does \label{sec:singlecore}
\seclabel{singlecore}

In the \defn{\efc} for the \singlecore, the emptier removes water from the fullest cup at each step. It was proven by Adler et al. \cite{AdlerBeFr03} that the \efc achieves backlog $O(\log n)$ at all times, independently of $\varepsilon \le 1/2$, and the lower bound \cite{BenderFeKr15, DietzRa91} shows that no deterministic algorithm can do better, regardless of the choice of $\varepsilon$. In this section, we consider the randomized version of the same problem, and present a randomized version of the \efc that achieves significantly better bounds.

\noindent\textbf{The Algorithm.}  At the start of the game, the emptier selects random \defn{starting states} $r_1, \ldots, r_n \in (0, 1]$, and places $r_j$ water in each cup $j$. (This can, of course, be simulated rather than actually performed.) Then, after each step of the algorithm, the emptier checks whether any cups have one or more units of water. If every cup has fill less than one, then the emptier does nothing. Otherwise, the emptier finds the fullest cup, and removes one unit of water from that cup.

We call our algorithm the \defn{\smoothefc}, because the algorithm can be seen as performing a variant of the deterministic \efc, except with the initial state of the cups smoothed so that each cup has a random quantity of water between zero and one. Surprisingly, this smoothing, which is only performed once at the beginning of the algorithm, results in significantly stronger bounds on the backlog at step $i$ for all $i \in \mathbb{N}$.

\begin{thm}
Suppose $\varepsilon \le 1/2$. Fix $i \in \mathbb{N}$, and consider some value $k \ge 0$ such that $3\log \frac{1}{\varepsilon} \le k$. Then the probability that the \smoothefc has backlog $k$ or greater after step $i$ is at most
$$O\left(\frac{1}{2^{2^{\Omega(k)}}} \right).$$

\label{thmloglog}
\end{thm}

By Theorem \ref{thmloglog}, for every step $i$, with high probability in $n$ the backlog does not exceed $O(\log \log n)$. Additionally, we will see that the \smoothefc maintains the bound of the deterministic \efc, guaranteeing deterministically that the backlog never exceeds $O(\log n)$.

\noindent\textbf{Analysis of Algorithm. }In order to prove Theorem \ref{thmloglog}, we begin by analyzing the deterministic \efc in a new setting. Consider a variant of the problem in which cups can be added and removed at each step. In particular, at the beginning of any step, arbitrarily many new cups may be added to the system, each initially containing no water. During the step, the filler must then place a positive amount of water in each of the new cups. A cup is removed from the system whenever it is completely emptied by the emptier. We call this variant of the game the \defn{dynamic \singlecore}. The following lemma extends the standard analysis of the deterministic \efc to this new game:

\begin{lem}
Consider the dynamic \singlecore, and suppose the emptier follows the \efc. If $m$ is the number of cups present after the $i$-th step in the game, then the fullest cup after the $i$-th step contains at most $O(\log m)$ water.
\label{lemdynamiccups}
\end{lem}

We defer the proof of Lemma \ref{lemdynamiccups} to Appendix \ref{appdynamic}.

Call a cup \defn{active} if the cup contains more than one unit of water, and \defn{inactive} otherwise. The \smoothefc essentially plays an instance of the dynamic \singlecore on the set of cups which are active. In particular, whenever a cup becomes active, it is because the filler has placed water overflowing the cup to contain more than one unit of water, and the only way a cup can become inactive is for the fill of the cup to be reduced to one or less. Since the \smoothefc always removes one unit of water from the fullest active cup, this corresponds to an instance of the dynamic \singlecore, where the cups in the new game are the active cups in the original game, and the fill of each cup in the new game is one less than the fill of the cup in the original game. Applying Lemma \ref{lemdynamiccups}, it follows that no cup contains more than $O(\log m)$ water, where $m$ is the number of active cups after a particular time step.

For the rest of the section, consider a fixed time step $i$. In order to bound the backlog after step $i$, it suffices to bound the number of active cups.

The difficulty in bounding the number of active cups directly is that there may be a series of time steps during which the emptier removes water exclusively from cups containing more than two units of water, thereby failing to make any active cups inactive, but the filler manages to increase the number of active cups significantly. In order to circumvent this issue, we instead bound a different quantity:

\begin{defn}
The \defn{integer fill} of the cups after a step $k$ is given by $\sum_{j \in [n]} \lfloor f_j(k) \rfloor$. We may also refer to the integer fill as the number of \defn{integer thresholds} passed after step $i$.
\end{defn}

Since each active cup passes at least one integer threshold, the integer fill is always at least as large as the number of active cups. Moreover, the integer fill is easier to bound due to the fact that the emptier always reduces the integer fill by one at the end of each step, unless the integer fill was already zero. The filler, on the other hand, may get lucky on some time steps and manage to cross a large number of integer thresholds using the $1 - \varepsilon$ units of water they pour. On average, though, the filler should not expect to cross more than $1 - \varepsilon$ thresholds per time step.

In order to bound the integer fill after the $i$-th step, we begin by proving a concentration bound for the number of integer thresholds crossed in any given sequence of $t$ consecutive steps. Note that consecutive steps are not independent. For example, if the filler places $1/2$ a unit of water in cup $j$ during a step, and then another $1/2$ a unit of water in cup $j$ during the next step, then exactly one of the two steps will cross an integer threshold in cup $j$. Despite this, it turns out that the number of thresholds crossed during any $t$ consecutive steps can still be expressed as a sum of independent zero-one random variables:

\begin{lem}
Consider a sequence of $t$ consecutive steps in the cup game in which the emptier follows the \smoothefc. Then the number of integer thresholds crossed during the $t$ steps can be expressed as a sum of independent zero-one random variables with total mean at most $(1 - \varepsilon) t$.
\label{lemconcentratedthresholds}
\end{lem}
\begin{proof}
Let $l, l + 1, \ldots, l + t - 1$ be the indices of the steps being considered. For any cup $j$, let $a_j$ denote the total amount of water placed in the cup by the filler prior to the $l$-th step, and let $b_j$ denote the amount of the water placed in the cup during steps $l, \ldots, l + t - 1$. That is, $a_j = \sum_{s < l} c_j(l)$ and $b_j = \sum_{l \le s < l + t} c_j(l)$. Recalling that $r_j$ units of water are poured into each cup $j$ at the beginning of the game, the total amount of water ever placed in the cup prior to step $l$ is $r_j + a_j$, and the total amount of water ever placed in the cup prior to step $l + t$ is $r_j + a_j + b_j$. The emptier only ever removes integer amounts from a cup, meaning that the number of integer thresholds crossed in the steps $l, \ldots, l + t - 1$ is independent of which steps the emptier targeted cup $j$. Indeed, the precise number of integer thresholds crossed in cup $j$ is given by
\begin{equation}
\lfloor r_j + a_j + b_j \rfloor - \lfloor r_j + a_j \rfloor.
\eqlabel{eqthresholds}
\end{equation}

The first $\lfloor b_j \rfloor$ units of water poured into the cup during the steps $l, l + 1, \ldots, l + t - 1$ will cross $\lfloor b_j \rfloor$ integer thresholds in total. Define $\overline{b_j} = b_j - \lfloor b_j \rfloor$ to be the fractional part of $b_j$. The final $\overline{b_j}$ units of water poured will cross an integer threshold if and only if

$$\lfloor r_j + a_j + \overline{b_j}\rfloor > \lfloor r_j + a_j \rfloor.$$

This, in turn, occurs if and only if there exists positive $q \le \overline{b_j}$ for which $r_j + a_j + q \equiv 0 \mod 1.$ Because $r_j + a_j \pmod 1$ is uniformly random (based on $r_j$), the probability of such a value $q$ existing is $\overline{a_j}$. Therefore, the number of integer thresholds crossed in cup $j$ can be written as a sum of indicator variables
$$X_1 + X_2 + \cdots + X_{\lfloor b_j \rfloor} + Y,$$
where each $X_i$ takes value $1$ with probability one, and where $Y$ is an indicator variable taking value $1$ with probability $\overline{a_j}$, depending on $r_j$. Because the $r_j$'s are independent between cups, the number of integer thresholds crossed is also independent between cups. Therefore, the total number of integer thresholds crossed during the steps $l, l + 1, \ldots, l + t - 1$ can be written as a sum of independent indicator variables, where the indicator variables for each individual cup have total expectation $a_j$, and thus the expected total sum of the indicator variables is $\sum_j a_j \le (1 - \varepsilon) t$.
\end{proof}

Applying Lemma \ref{lemconcentratedthresholds}, we can bound the probability for any $k$ that there are $k$ or more active cups after step $i$.

\begin{lem}
Consider $k \ge 1 / \varepsilon^3$. With probability at least $1 - e^{-\Omega(k^{1/3})},$ the integer fill after step $i$ is no more than $k$.
\label{lemboundedthresholds}
\end{lem}
\begin{proof}
Let $t \ge 0$ be the smallest $t$ such that at the end of step $i - t$, the integer fill was zero (i.e., there were no active cups). Then during each of the $t$ steps following step $i - t$, the emptier must have successfully removed one unit of water from some cup, thereby reducing the integer fill by one. Hence in order for the integer fill after step $i$ to be at least $k$, at least $t + k$ integer thresholds must have been crossed in the preceding $t$ steps.

Let $S_t$ denote the number of integer thresholds crossed in steps $i - t + 1, \ldots, i$. In order to complete the proof, it suffices to bound the probability that $S_t \ge t + k$ for any $t$. By Lemma \ref{lemconcentratedthresholds}, $S_t$ can be expressed as a sum of independent indicator variables, and $\E[S_t] = (1 - \varepsilon)t$. By a Chernoff bound, it follows that
\begin{equation}
\Pr[S_t \ge \E[S_t] + \delta \E[S_t]] \le e^{-\delta^2 \E[S_t] / 3},
\eqlabel{eqchernoff1}
\end{equation}
for $\delta \le 1$, and
\begin{equation}
\Pr[S_t \ge \E[S_t] + \delta \E[S_t]] \le e^{-\delta \E[S_t] / 3},
\eqlabel{eqchernoff2}
\end{equation}
for $\delta \ge 1$. Applying \eqref{eqchernoff2} in the case of $k \ge t$, we get
\begin{equation}
\Pr[S_t \ge t + k] \le \Pr[S_t \ge \E[S_t] + k] \le e^{-k / 3}.
\eqlabel{eqchernoff3}
\end{equation}
Applying \eqref{eqchernoff1} in the case of $k \le t$, we get
\begin{equation}
\Pr[S_t \ge t + k] \le \Pr[S_t \ge \E[S_t] + \varepsilon t] \le e^{-\varepsilon^2 \E[S_t] / 3} \le e^{-\varepsilon^2 t / 6},
\eqlabel{eqchernoff4}
\end{equation}
where the final inequality uses the fact that $\E[S_t] = (1 - \varepsilon)t \ge t/2$.
Summing over all $t$, we get that
\begin{align*}
\Pr[S_t \ge t + k \text{ for some }t] &\le \sum_{t \ge 1} \Pr[S_t \ge t + k] \\
	      	  	     	   &\le  k \cdot e^{-k/3} + \sum_{t \ge k}  e^{-\varepsilon^2 t / 6} \\
                                   & =  k \cdot e^{-k/3} +    e^{-\varepsilon^2 k/6} \cdot\sum_{t \ge 0}  \left(e^{-\varepsilon^2 / 6}\right)^t \\
                                   & =  k \cdot e^{-k/3} +    e^{-\varepsilon^2 k/6} \frac{1}{1 - e^{-\varepsilon^2 / 6}}. \\
				   & \le  k \cdot e^{-k/3} +    e^{-\varepsilon^2 k/6} O(1/\varepsilon^2). \\
				   &\le  O\left(e^{-\varepsilon^2 k/6} / \varepsilon\right). \\
\end{align*}
Assuming $k \ge 1/\varepsilon^3$, this becomes $e^{-\Omega(k^{1/3})}.$ Therefore with probability at least $1 - e^{-\Omega(k^{1/3})},$ the integer fill after step $i$ is no more than $k$. 
\end{proof}

Putting the pieces together, we	can complete the proof of Theorem \ref{thmloglog}.
\begin{proof}[Proof of Theorem \ref{thmloglog}]
Consider $k$ such that $3\log \frac{1}{\varepsilon} \le k$, meaning that $1/\varepsilon^3 \le 2^k$. By Lemma \ref{lemboundedthresholds}, with probability at least $1 - e^{-\Omega(2^{k/3})}$, the integer fill after step $i$ is at most $2^k$. Since the integer fill is at least as large as the number of active cups, it follows that there are at most $2^k$ active cups. Applying Lemma \ref{lemdynamiccups} to the active cups, we get that the backlog is at most $O(k)$, as desired.
\end{proof}

%!TEX root =  main.tex

\section{Deterministic Results for the \MultiCore}
\label{sec:multicoredeterministic}

Recall that the \multicore works as
follows. At each step the filler is allowed to distribute up to $(1 -
\epsilon)p$ units of water among the cups, placing up to $1 - \delta$
units in any individual cup; and then the emptier is permitted to
remove up to one unit of water from each of up to $p$ distinct
cups. As in the \singlecore, the emptier's goal is
to minimize backlog.

In this section we consider the \multicore in the deterministic
setting, in which the emptier's strategy is to simply always remove (up
to) one unit of water from the $p$ fullest cups (also known as the
\efc). Until now, no nontrivial upper bounds for the maximum
backlog have been proven in this setting. We achieve a maximum backlog of $O(\log n)$, when the emptier
is allowed to remove a constant-fraction more water in total than the
filler pours, and at least a $1 / \poly(n)$ amount more from
each individual cup (i.e., $\epsilon \ge \Omega(1)$ and $\delta \ge
1 / \poly(n)$).

The goal of the section will be to prove Theorem \ref{thmmulticoredeterministic}:
\begin{thm}
Let $0 < \epsilon < 1$ and $\delta \ge \frac{1}{\poly(n)}$.  Then the \efc will
achieve maximum backlog $O(\frac{1}{\epsilon} \log n)$ on a game with
$n$ cups.
\label{thmmulticoredeterministic}
\end{thm}

To introduce the approach for proving Theorem
\ref{thmmulticoredeterministic}, we begin by revisiting the
single-processor \flushgame \cite{DietzRa91}.

\subsection{Revisiting the \FlushGame}

Recall that in each step of the \defn{\flushgame}, the filler
distributes one unit of water among the cups, and then the emptier
removes \emph{all} of the water from some cup. A result that has found
numerous applications in deamortization is that, if the emptier
follows the \efc, then no cup will ever have fill more than $O(\log
n)$
\cite{AmirFaId95,DietzRa91,DietzSl87,AmirFr14,Mortensen03,GoodrichPa13,FischerGa15,Kopelowitz12}. In
this subsection, we present a new proof of the result. In additional
to bringing new intuition to the problem, our proof has the
interesting property that it can easily be adapted to work even if the
emptier does not remove from the fullest cup at each step, but instead
removes from a cup whose fill is within $O(1)$ of the maximum.

Define $f_j(i)$ to be the amount of water in cup $j$ after step
$i$. The key insight in the proof is to examine the potential
function,
$$\phi(i) = \sum_{j = 1}^n (1.1)^{f_j(i)},$$ obtained by
exponentiating the water in each cup, and then taking a sum. We will
show that $\phi(i)$ never exceeds $O(n)$, by proving that at the end
of each step $i$, either $\phi(i) < \phi(i - 1)$, or no cup contains
more than $2$ units of water (and thus $\phi(i) \le O(n)$). The bound
on $\phi(i)$ prevents any single cup from ever having fill greater
than $O(\log n)$.

For convenience, we will actually use a
slightly different potential function,
$$\phi(i) = \sum_{j = 1}^n \int_{x = 0}^{f_j(i)} 1.1^{\lceil x \rceil} dx.$$

Using this potential function, the proof proceeds as follows. Consider
a given step $i$. Let $t$ be the fill of the fullest cup after the
filler takes their turn. Assume that $t \ge 2$, since otherwise no cup
contains more than $2$ units of water, and $\phi(i)$ cannot exceed
$O(n)$. The water added to the cups can have increased the potential
by at most $1.1^{\lceil t \rceil}$. On the other hand, the emptier
will remove at least $2$ units of water from the fullest cup, thereby
decreasing the potential by at least
\begin{equation}
  1.1^{\lfloor t \rfloor} + 1.1^{\lfloor t \rfloor - 1} \ge 1.1^{\lceil t \rceil}.
  \eqlabel{eqlooset}
\end{equation}
Thus the potential after step $i$ satisfies $\phi(i) \le \phi(i - 1)$.

Since either $\phi(i) \le O(n)$ or $\phi(i) \le \phi(i - 1)$ for each
$i > 0$, and since $\phi(0)$ is trivially $0$, the potential $\phi(i)$
must always remain in $O(n)$. This, in turn, implies that $f_j(i) \le
O(\log n)$ for all cups $j$ and steps $i$.

\subsection{Proving Theorem \ref{thmmulticoredeterministic}}

The proof of Theorem \ref{thmmulticoredeterministic} generalizes the
ideas from the previous section.

\begin{proof}[Proof of Theorem \ref{thmmulticoredeterministic}]
  Define $f_i(j)$ to the fill in cup $j$ after step $i$, and define the potential function,
  $$\phi(i) = \sum_{j = 1}^n \int_{x = 0}^{f_j(i)} (1 + \epsilon)^{\lceil x \rceil} dx.$$

  We will prove that $\phi(i) \le O(\poly(n))$ for all $i$. This, in
  turn, bounds each $f_i(j)$ to be at most $O(\log_{1 + \epsilon}
  \poly(n)) \le O(\frac{1}{\epsilon} \log n)$.

  For a given step $i$, we consider two cases separately:

  \noindent \textbf{Case 1: The emptier removes a full unit of water from each of $p$ cups during step $i$. } In this case, we break the
  emptying in step $i$ into three substeps, beginning after the filler
  has taken their turn: (1) The emptier identifies the $p$ fullest
  cups $r_1, \ldots, r_{p}$; (2) The emptier removes from each of the
  cups $r_i$ any water that was poured into it during step $i$; (3)
  The emptier removes additional water so that a total of one unit has
  been removed from each of the cups $r_1, \ldots, r_p$. Note that the
  case requirement ensures that each of the cups $r_1, \ldots, r_p$
  contains at least a unit of water during substep (1), and thus substep (3)
  is well defined.

  Let $h$ denote the fill of the emptiest cup in the set $\{r_1,
  \ldots, r_{p}\}$ at substep (1). Let $l$ be the number of units of
  water that the filler places in cups during step $i$ and that still
  remains after substep (2). (This corresponds with water placed in cups
  besides $r_1, \ldots r_p$.) The total contribution of the $l$ units
  of water to the potential can be at most
  $$l \cdot (1 + \epsilon)^{\lceil h \rceil}.$$

  The amount of water removed in substep (2) is at most $p \cdot (1 -
  \epsilon) - l = p - (l + \epsilon p)$. Thus at least $l + \epsilon
  p$ further units of water are removed during substep (3). The removal
  of water that occurs in substep (3) decreases the potential by at least
  $$(l + \epsilon p) (1 + \epsilon)^{\lceil h \rceil - 1} \ge (1 + \epsilon) \cdot l \cdot (1 + \epsilon)^{\lceil h \rceil - 1} \ge l \cdot (1 + \epsilon)^{\lceil h \rceil}.$$

  Thus the potential $\phi(i)$ at the end of step $i$ satisfies
  $\phi(i) \le \phi(i - 1)$.

  \noindent \textbf{Case 2: At no point during step $i$ do $p$ or more cups contain one or more units of water.} Call
  the cups that contain more than one unit of water at the end of the step the \textbf{heavy hitters}. Taking a very loose bound,
  any water placed in non-heavy hitters during step $i$ can increase
  the potential by at most $2 \cdot n$ in total. On the other hand,
  any water that is placed in a heavy hitter during step $i$ is then immediately removed by the
  emptier. Additionally, the emptier removes from each heavy
  hitter at least $\delta$ units of water that were present at the beginning
  of the step. If the fullest heavy hitter has fill $t$ at the
  beginning of the step, then this removal will reduce the potential
  by at least
  $$(1 + \epsilon)^{\lfloor t \rfloor} \cdot \delta.$$ Recall that
  $\delta \ge \frac{1}{\poly(n)}$. Thus if $t$ is sufficiently large
  in $\Omega(\log_{1 + \epsilon} n)$, then it follows that the
  $\delta$-unit removal of water will decrease the potential by at
  least $2n$. This forces the net change in potential $\phi(i) -
  \phi(i - 1)$ in step $i$ to be negative.

  If, on the other hand, $t \le O(\log_{1 + \epsilon} n)$, then the
  total potential $\phi(i)$ at the end of step $i$ can be at most
  $$n \cdot \int_{x = 0}^{t} (1 + \epsilon)^{\lceil x \rceil} dx \le \poly(n).$$

  In both Case (1) and Case (2), at the end of each step $i$, the
  total potential is either bounded by $\poly(n)$, or satisfies
  $\phi(i) \le \phi(i - 1)$. It follows that the potential $\phi(i)$
  is always bounded by $\poly(n)$, completing the proof.
\end{proof}

%!TEX root =  main.tex

\section{A Randomized Algorithm for the \MultiCore with Constant Backlog}
\seclabel{multicoreconstantbacklog}

%% Unlike for the single-core version of the problem, there is no deterministic algorithm known providing any bounds on the backlog. \todo{Of course, this is no longer true now... :P} It is unknown, for example, whether a sufficiently clever filler might be able to achieve backlog as a function of $p$ with constant probability. In this section, we present an algorithm showing that this is not the case. In fact, we show that when the number of processors $p$ is reasonably large, and for $\epsilon$ and $\delta$ that are not too small, there is an algorithm that with good probability keeps the backlog \emph{at most constant} at any given step.

In this section, we present the first evidence that the emptier in the
\multicore is actually \emph{more powerful} than the filler in the
\singlecore. In particular, we design a randomized
algorithm for the emptier which allows for them to guarantee with large
probability (in $p$) at any given step that the backlog does not
exceed $O(1)$. The techniques we introduce here will also play a
critical role in Section \ref{sec:multicoresmoothefc} where we we extend the
algorithm to provide probabilistic guarantees against any backlog $k$.

Throughout the section we use the \tweakedemphmulticore, in which at
each step the filler distributes $(1 - \epsilon)$ units of water among
cups, placing no more than one unit in any cup; and the emptier then
removes up to $(1 + 2\delta)$ units of water from each of up to $p +
1$ cups. The goal of the section is to prove the following theorem:
\begin{thm}
Suppose $1/2 \ge \varepsilon \ge p^{-1/3}$ and $\delta < 1$ is sufficiently large in $\Omega(e^{-O(\varepsilon^2p)})$. Then there exists an algorithm for the \tweakedmulticore that guarantees after any given step $i$ that with probability at least $1 - O\left(e^{-\Omega(\varepsilon^2p)}\right)$ the backlog does not exceed $3$. 
\label{thmmulti-core} 
\end{thm}

The design and analysis of the algorithm for Theorem \ref{thmmulti-core} are somewhat more complex than in the single-processor case, and the analysis, in particular, invokes a concentration inequality for Doob martingales. We begin with an algorithm for a simpler version of the problem. 

\noindent \textbf{A Warmup: Games of Bounded Length} We begin by proving Theorem \ref{thmmulti-core} for the special case where the game terminates after $m$ steps for some $m \le e^{\epsilon^2 p / 12}$. The algorithm for this case is particularly simple, and motivates the algorithm for the general case. (Moreover, note that the algorithm in this case works even for $\delta = 0$.)

At the start of the game, the emptier selects random \defn{starting states} $r_1, \ldots, r_n \in (0, 1]$, and places $r_j$ water in each cup $j$. After each step, the emptier removes one unit of water from every cup containing one or more units, unless there are more than $p + 1$ such cups, in which case the emptier selects $p + 1$ of them arbitrarily. 

Suppose that at the beginning of some step $i$, no cups have one or more units of water.  Let $c_j(i)$ denote the amount of water which the filler places in each cup $j$. By the same reasoning as in the single-processor case, each cup will independently cross an integer threshold with probability $c_j(i)$. Consequently, the total number $T$ of cups that cross a threshold to contain one or more units of water will be a sum of indicator random variables with total mean at most $(1 - \epsilon)p$. By a Chernoff bound (applied to $\epsilon < 1/2$), 
$$\Pr[T > p] \le e^{-\epsilon^2 (1 - \epsilon) p / 3} \le e^{-\epsilon^2 p / 6}.$$
Moreover, as long as $T \le p + 1$, then step $i$ will end with no cups having one or more units of water. Applying a union bound over the $m$ steps, the probability that there is any step $i$ after which some cup contains one or more units of water is at most
$$m \cdot e^{-\epsilon^2 p / 6} \le e^{-\epsilon^2 (1 - \epsilon) p / 12}.$$

Thus we have the following version of Theorem \ref{thmmulti-core}:
\begin{thm}
Suppose $\varepsilon \le 1/2$, and suppose that the \multicore is played for $m \le e^{\epsilon^2 p/12}$ steps. Then there exists an algorithm which guarantees with probability $1 - e^{-\Omega(\varepsilon^2p)}$ that the backlog never exceeds $1$ at the end of any step.
\label{thmmulti-corewarmup}
\end{thm}

\begin{remark}
  One can also consider the case where $\epsilon > 1/2$. Applying the same argument, except with the appropriately modified Chernoff bounds, the probability $\Pr[T > p]$ becomes bounded by $e^{-(1 - \epsilon) p / 3}$. By a union bound, the probability that the backlog ever exceeds $1$ at the end of any of the first $m$ steps, becomes at most $e^{-(1 - \epsilon) p/3} / m$. Thus one can use larger $\epsilon$ values in order to achieve better constants in the exponent of the probability bound, which may be useful for some applications where $p$ is small. Although we do not do it out explicitly, the use of larger values of $\epsilon$ has the same effect on Theorem \ref{thmmulti-core}.
\end{remark}

\noindent \textbf{Generalizing to Games of Unbounded Length } The algorithm for the bounded-length case can be reframed as follows. For each cup $j$, a random value $s_j$ is selected (equal to $1 - r_j$ in the original version of the algorithm). Additionally, the emptier maintains a counter $w_j$ for each cup. Whenever the total amount of water ever poured into a cup $j$ crosses a threshold of the form $k + s_j$ for some $k \in \mathbb{N}$, the counter $w_j$ is incremented. The emptier's algorithm at each step is then to simply select (up to) $p + 1$ cups with non-zero counters, to remove (up to) one unit of water from each of those cups, and then to decrement each of the corresponding counters.

If the filler never manages to cross thresholds in more than $p + 1$ cups in any step, then all of the counters will be maintained at zero at the end of each step, and no cup will contain more than one unit of water. This is what happens in the bounded-length case. In the unbounded-length case, we must handle the fact that the filler may sometimes manage to make a large number of cups cross thresholds at the same time. Handling this requires not only a more sophisticated analysis, but also a noteworthy change to the algorithm design. In particular, the random value $s_j$ associated with cup $j$ changes over time, that way the filler cannot correctly guess $s_j$ at one point in time and then use it in order to make cup $j$ misbehave indefinitely by crossing thresholds at inconvenient times for the emptier. The generalized algorithm is presented next.
%% The second is that the emptier must devote some of its recourses to running a background process which attempts to clean up any damage from previous steps in which too many cups crossed a threshold at once.

\noindent \textbf{The Algorithm }
We assume without loss of generality that $\frac{1}{\delta}$ is a natural number. Set the threshold $r_j(0)$ to be null in each cup $j$ (i.e., non-existent). Then for every $i \in \mathbb{N}$ such that $i \equiv 1 \mod 1/\delta + 1$, and for every cup $j$, set the threshold $r_j(i)$ to be null (i.e., non-existent), then select a random value $s_j(i) \in (0, 1]$, and set the thresholds $r_{j}(i + 1) = (i + 1) + s_j(i), \ldots, r_j(i + 1/\delta) = (i + 1/\delta) + s_j(i)$. Let $a_j(i)$ denote the total amount of water poured into cup $j$ during the first $i$ steps. We say that cup $j$ \defn{crosses a threshold} at step $i$ whenever there is some $r_{j}(t)$ such that $a_j(i - 1) < r_j(t) \le a_j(i)$.

Our definition of the thresholds $r_j(i)$ has three important properties: The first is that within a given cup, each threshold $r_j(i)$ is dependent on at most $1/\delta$ other thresholds $r_j(i')$. This independence between far-apart thresholds will be important in the analysis when we wish to prove concentration bounds concerning threshold crossings. The second property is that no two thresholds $r_j(i)$ and $r_j(i')$ are ever within less than one of each other. Consequently, no cup can ever cross more than a single threshold in a given step. The third property is that the threshold $r_j(i)$ is only null for one in every roughly $1 / \delta$ values of $i$. This will ensure that as long as $1 + \delta$ units of water are removed from a cup each time a threshold is crossed, then the cup will be kept at a steady state of containing only a constant fill after each step. 

For every cup $j$, the algorithm maintains a \defn{threshold counter} $w_j$, initially set to zero, which is incremented by $1 + \delta$ every time the cup crosses a threshold. The counter is reduced by the emptier (but never below $0$) as follows: whenever the emptier empties from a cup $j$, they remove $t = \min(1 + 2 \delta, w_j)$ units of water from the cup and reset $w_j = w_j - t$.

We will show that the spacing of thresholds ensures the total quantity the filler has ever added to $w_j$ is at most the total amount of water that has been poured into cup $j$; since the emptier removes from $w_j$ the same amount as they remove water from the cup, it follows that there is always at least $w_j$ water in each cup. This means that the emptier's strategy of removing $\min(1 + 2\delta, w_j)$ units of water is well defined.

Using the threshold counters, the algorithm through which the emptier selects cups is simple. At each step the emptier may select an arbitrary set of $p + 1$ cups to empty from, as long as priority is given to cups $j$ satisfying $w_j \ge 1 + \delta$ (over all other cups) and to cups satisfying $w_j \neq 0$ (over cups satisfying $w_j = 0$). The emptier then removes $\min(w + 2\delta, w_j)$ water from each of the selected cups $j$, and updates the threshold counters appropriately.

\noindent \textbf{Algorithm Analysis }We begin by showing that each of
the counters $w_j$ closely tracks the backlog in the corresponding cup
$j$. As a convention, we will often use $w_j(i)$ as a shorthand to
denote the value of the counter $w_j$ after step $i$. 

\begin{lem}
  The amount of water $f_j(i)$ contained in cup $j$ after step $i$
  satisfies $w_j(i) \le f_j(i) \le w_j(i) + 3$.
  \label{lemcounterbacklog}
\end{lem}
\begin{proof}
  Let $l$ be the number of thresholds crossed in cup $j$ during the
  first $i$ steps. Let $b$ be the net amount of water removed by the
  emptier during those steps. Then $w_j(i) = l \cdot (1 + \delta) -
  b$. It therefore suffices to compare $l \cdot (1 + \delta)$ to the
  total amount of water $k$ poured into $w_j$ in the first $i$
  steps. Namely, we wish to show that $l \cdot (1 + \delta)
  \le k \le l \cdot (1 + \delta) + 3$.

The total amount of water poured into cup $j$ can be broken into the
sum of three quantities $k = k_0 + k_1 + k_2$, where $k_0 = 1$, $k_1$
is some multiple of $\frac{1}{\delta} + 1$, and $k_2 <
\frac{1}{\delta} + 1$.\footnote{Note that the case in which less than one unit of water has been placed in cup $j$ is trivially handled on its own.} Once $k_0 + k_1$ units of water has been poured into
the cups the number $l_1$ of thresholds crossed will exactly satisfy
\begin{equation}
l_1 \cdot (1 + \delta) = k_1
\label{eq:k1}
\end{equation}
When $k_2$ additional units are added,
the number $l_2$ of additional thresholds crossed will satisfy $k_2 -
2 \le l_2 \le k_2$. Since $l_2 \cdot (1 + \delta) \le l_2 + 1$, it
follows that
\begin{equation}
  k_2 - 2 \le l_2 \cdot (1 + \delta) \le k_2 + 1 = k_2 + k_0
  \label{eq:k02}
\end{equation}

Combining \eqref{k1} and \eqref{k02}, we get that
$$k_1 + k_2 - 2 \le l \cdot  (1 + \delta) \le k_0 + k_1 + k_2,$$
which in turn implies that $k - 3 \le l \cdot (1 + \delta) \le k$.
This rearranges to $l \cdot (1 +
\delta) \le k \le l \cdot (1 + \delta) + 3$, as desired.
\end{proof}

Lemma \ref{lemcounterbacklog} tells us that, in order to bound the
backlog by $3$ after step $i$, it suffices to demonstrate that $\sum_j
w_j(i) = 0$. Next we turn to the task of proving concentration bounds
on $\sum_j w_j$. If $k$ thresholds are crossed by water poured during
a given step $i$, we define the random variable $T_i = k - \min(k, p)$
to be the number of \defn{surplus thresholds} crossed during the
step. A key insight is that for any $t \ge 0$, we can express
$\Pr[\sum_j w_j(i) > t]$ as a statement about sums of $T_i$'s. Namely,
we can show that if $\sum_j w_j(i)$ is large, then there must be some
$l$ such that the number of surplus thresholds crossed in the $l$
steps leading up to step $i$ was also large.

\begin{lem}
  Call a positive integer $l$ a \defn{height-$t$ backlog witness for step $i$} if
$2\sum_{m =  i - l + 1}^i T_m \ge \delta l + t.$  For any $t \ge 0$ and any step $i$, if $\sum_j w_j(i) > t$,
  then there exists a height-$t$ backlog witness for step $i$.
  \label{lemboundwsum}
\end{lem}
\begin{proof}
  Suppose $\sum_j w_j(i) > t$, and define $l > 0$ to be the smallest
  $l \ge 0$ such that $\sum_j w_j(i - l) = 0$. (Note that such an $l$ exists
  since $l = i$ will always work.)

  We claim that during each step $m \in \{i - l + 1, \ldots, i\}$, if
  $k$ is the number of thresholds crossed, then the emptier reduces
  the counters by at least
  $$\min(k, p) \cdot (1 + \delta) + \delta.$$

  In particular, during the step $m$, $k$ counters will be incremented
  by $1 + \delta$. It follows that the emptier will be able to
  successfully be able to reduce at least $\min(k, p + 1)$ counters by
  at least $1 + \delta$ each.

  If $k \ge p + 1$, then the total reduction in the counters will therefore be
  at least
  $$(p + 1) \cdot (1 + \delta) \ge \min(k, p) \cdot (1 + \delta) + \delta.$$

  If, on the other hand, $k \le p$, then the emptier will have been
  able to successfully reduce $k$ counters by at least $1 + \delta$
  each. If this is the full extent to which the emptier is able to
  reduce the counters (i.e., the emptier reduces $\sum_j w_j$ by
  exactly $(1 + \delta)k$ in this step), then all $k$ of the reduced
  counters must have had values $w_j$ exactly equal $1 + \delta$ prior
  to the reduction, and all other counters $w_j$ must have already
  been zero. This cannot occur, however, since we know that $\sum_j
  w_j(m)$ is positive at the end of step $m$ (due to the fact that $m$
  comes after step $i - l$). Thus we may conclude that during step $m$
  the emptier reduces the $\sum w_j$ by more than $k \cdot (1 +
  \delta)$. As an invariant, every counter always remains a multiple
  of $\delta$, meaning that the emptier must actually reduce $\sum_j
  w_j$ by at least $k \cdot (1 + \delta) + \delta$.

  We have seen that during step $m$, regardless of whether $k > p$ or
  $k \le p$, the emptier reduces the counters by at least
  $$\min(k, p) \cdot (1 + \delta) + \delta.$$
  
  On the other hand, the filler increased the counters by
  $$(1 + \delta) \cdot k = (1 + \delta) \cdot (T_m + \min(k, p)) =
  \min(k, p) \cdot (1 + \delta) + T_m \cdot (1 + \delta).$$

  Thus the net increase in $\sum_j w_j$ during step $m$ is at most
  $$T_m \cdot (1 + \delta) - \delta.$$

  Over the course of the $l$ steps $i - l + 1, \ldots, i$, the sum
  $\sum_j w_j$ must increase in total by more than $t$. It follows that 
  $$\sum_{m =  i - l + 1}^i (T_m \cdot (1 + \delta) - \delta) > t.$$
  Hence
  \begin{equation}
    2\sum_{m =  i - l + 1}^i T_m  > \delta l + t,
    \eqlabel{eqlwitness}
  \end{equation}
  which establishes that $l$ is a height-$t$ backlog witness.
\end{proof}

By Lemma \ref{lemboundwsum}, if we wish to bound $\Pr[\sum_j w_j > 0]$
at a step $i$, then it suffices to instead bound the probability that
there exists a height-0 backlog witness (which, for brevity, we will
often refer to simply as a backlog witness). In order to do this, we
begin by bounding the probability that $T_k > 0$ for any given $k$.

\begin{lem}
For a given $k$, $\Pr[T_k > 0] \le e^{-\varepsilon^2 p / 6}$.
\label{lemmulti-coresinglestep}
\end{lem}
\begin{proof}
For each cup $j$, if $c_j(k)$ units of water are poured into the cup during step $k$, then the probability of crossing a threshold in that cup is at most $c_j(k)$, depending on the value of the appropriate threshold. Since thresholds are independent between cups, the number of thresholds crossed in step $k$ is a sum of independent indicator variables with total mean at most $(1 - \varepsilon)p$. By a Chernoff bound, the probability that $m$ or more thresholds are crossed is therefore at most
$$e^{-\varepsilon^2 (1 - \varepsilon)p / 3} \le e^{-\varepsilon^2 p / 6},$$
where the final inequality uses that $\varepsilon \le 1/2$.
\end{proof}

It will also be useful to have a bound for $\E[T_k]$.
\begin{lem}
$\E[T_k] \le O(p \cdot e^{-\varepsilon^2 p / 6})$.
\label{lemexpTi}
\end{lem}
\begin{proof}
Recalling that the number of thresholds crossed in step $k$ is a sum of independent indicator variables with total mean at most $(1 - \varepsilon)p$,
\begin{align*}
\E[T_k]& \le \Pr[T_k > 0] \cdot p + \E[T_k \cdot \mathbb{I}(T_k > p)]\\
       & \le p e^{-\varepsilon^2p / 6} + \sum_{m > p} m \cdot \Pr[T_k > m]\\
       & \le p e^{-\varepsilon^2p / 6} + \sum_{m > p} m \cdot e^{-m / 3}\\
       & \le p e^{-\varepsilon^2p / 6} + O(p \cdot e^{-p/3})\\
       & \le O(p e^{-\varepsilon^2p / 6}).
\end{align*}
\end{proof}

Now we return to our task of bounding the probability of there existing a backlog witness $l > 0$. Define $M = e^{\varepsilon^2p / 12}$. By Lemma \ref{lemmulti-coresinglestep} and the union bound, with probability at least $1 - 1/M$ we have $T_i, T_{i - 1}, \ldots, T_{i - M + 1} = 0$, and thus that $1, 2, \ldots, M$ are not backlog witnesses.

Next we consider the probability that $l$ is a backlog witness for $l > M$. Unfortunately, $T_i + T_{i - 1} + \cdots + T_{i - l + 1}$ cannot easily be expressed as a sum of independent random variables. In particular, because each $T_i$ can only be made to be greater than zero by a combination of events in multiple cups, one cannot perform a per-cup analysis in the same way as for the single-processor case.

Instead, we will employ a more sophisticated analysis in order to
prove the following proposition:
\begin{prop}
  Suppose $1/2 \ge \varepsilon \ge p^{-1/3}$ and $\delta < 1$ is
  sufficiently large in $\Omega(e^{-O(\varepsilon^2p)})$. For any $l >
  0$, and any step $i$, the probability that $l$ is a height-$t$
  backlog witness for step $i$ is at most
  $$O(pl) \cdot \exp\left(-\Omega\left(\left(\frac{\delta^2 \left(\delta l +
      t\right)^2}{lp}\right)^{1/3}\right)\right).$$
  \label{propwitnessprob}
\end{prop}

Before presenting a proof of Proposition \ref{propwitnessprob}, we
first use the proposition in order to complete the proof of Theorem
\ref{thmmulti-core}.

\begin{proof}[Proof of Theorem \ref{thmmulti-core}]
  By Lemma \ref{lemcounterbacklog}, it suffices to examine the probability that $\sum_j w_j(i) = 0$. This is guaranteed to occur whenever there are no backlog witnesses $l \ge 1$ for step $i$.

  We have already shown that the probability of any $l \le M$ being a
  backlog witness is at most $1/M$. Now consider some $l \ge M$. By
  Proposition \ref{propwitnessprob} the probability of $l$ being a
  backlog witness is at most
$$O(pl) \cdot \exp\left(-\Omega\left(\left(\frac{\delta^4
    l}{p}\right)^{1/3}\right)\right).$$ Using the fact that $\delta$
  is sufficiently large in $\Omega\left(e^{-O(\varepsilon^2
    p)}\right)$, we may assume that $\delta \ge e^{p^{1/3} /
    120}$. Using the fact that $l \ge M \ge e^{p^{1/3} / 12} \ge
  \Omega(p^{10})$, and that $l \ge M \ge 1/\delta^{10}$, it follows
  that the probability of $l$ being a backlog witness is at most
\begin{align*}
O(pl) \cdot e^{-\Omega\left(l^{1/6}\right)} \le O\left(e^{-\Omega\left(l^{1/6}\right)}\right) \le O\left(\frac{1}{l^2}\right).
\end{align*}
By the union bound, the probability of there being any any $l \ge M$ that is a backlog witness is at most
$$O\left(\sum_{l = M - 1}^\infty \frac{1}{l^2}\right) = O\left(\int_{M - 1}^\infty \frac{1}{l^2} dl \right) = O(1/M).$$

Recall that with probability at least $1 - 1/M$, no value of $l \le M$
is a backlog witness. It follows that with probability $1 - O(1/M)$
there are no backlog witness values of $l$. Since $M =
e^{\varepsilon^2p / 12}$, this completes the proof.
\end{proof}

The remainder of the section is devoted to proving Proposition
\ref{propwitnessprob}.

\subsection{Proof of Proposition \ref{propwitnessprob}}

An essential ingredient to the proof will be the use of McDiarmid's Inequality.

\begin{thm}[McDiarmid's Inequality \cite{McDiarmid89}]
Let $X_1, \ldots, X_m$ be independent random variables over an arbitrary probability space. Let $F$ be a function mapping $X_1, \ldots, X_m$ to $\mathbb{R}$, and suppose $F$ satisfies,
$$\sup_{x_1, x_2, \ldots, x_n, \overline{x_i}} |F(x_1, x_2, \ldots, x_{i - 1}, x_i, x_{i + 1}, \ldots , x_n) - F(x_1, x_2, \ldots, x_{i - 1}, \overline{x_i}, x_{i + 1}, \ldots , x_n)| \le c,$$
for all $1 \le i \le n$. That is, if $X_1, X_2, \ldots, X_{i - 1}, X_{i + 1}, \ldots, X_n$ are fixed, then the value of $X_i$ can affect the value of $F(X_1, \ldots, X_n)$ by at most $c$. Then for all $R > 0$,
$$\Pr[F(X_1, \ldots, X_n) - \E[F(X_1, \ldots, X_n)] \ge R] \le e^{-2R^2 / (c^2n)},$$
and 
$$\Pr[F(X_1, \ldots, X_n) - \E[F(X_1, \ldots, X_n)] \le -R] \le e^{-2R^2 / (c^2n)}.$$
\end{thm}

McDiarmid's inequality can be viewed as a special case of Azuma's inequality applied a Doob martingale \cite{AlonSp04}. In particular, one can obtain McDiarmid's inequality by applying Azuma's inequality to the martingale $(B_0 , \ldots, B_n)$ with $B_i = \E[F(X_1, \ldots, X_n) \mid X_1, \ldots, X_i]$. We will be using McDiarmid's inequality directly, without explicitly constructing the corresponding martingale.

%% McDiarmid's inequality can be useful in situations where standard Chernoff bounds fail to apply.  For instructional purposes, we present one simple example of this below. For additional examples, see \cite{??}.

%% \begin{example}
%% Suppose we place $n$ balls randomly into $n$ bins, with each ball placed independently of the others. Let $X$ denote the number of bins not containing any balls. It is easy to see that $\E[X] \approx n/e$. Proving a concentration inequality for $X$ is somewhat more subtle, since the probability of one bin being empty is influenced by how many of the other bins are empty. Notice, however, that $X$ can be written as a function $F(X_1, \ldots, X_n)$, where $X_i$ represents the bin in which ball $i$ was placed. Moreover, since a single ball can change the number of empty bins $X$ by at most one, one can apply McDiarmid's Inequality in order to conclude that $X$ is tightly concentrated around its mean.
%% \end{example}

In our application of McDiarmid's inequality, we will partition the values $s_j(k)$ (i.e., the random values used in our algorithm) into a partition $Q_1, \ldots, Q_{pl + 1}$. Roughly speaking, each random variable $X_t$ will represent the values of all the $s_j(k)$'s in the component $Q_t$ of the partition. The key technical difficulty is to design the partition $Q_1, \ldots, Q_{pl + 1}$ so that the condition for McDiarmid's inequality is met, thereby allowing us to obtain a concentration inequality on the value of $T_i + T_{i - 1} + \cdots + T_{i - l + 1}$.

For a cup $j$, and a non-negative $k$ satisfying $k \equiv 1 \mod 1/\delta + 1$, define the \defn{threshold collection $(j, k)$} to be the set of thresholds $r_j(k + 1), \ldots, r_j(k + 1/\delta)$. In particular, these are the thresholds whose values are determined by $s_j(k)$. We say that step $t$ is \defn{hit by threshold collection $(j, k)$} if during step $t$, the total amount of water ever poured into cup $j$ crosses one of the thresholds in the collection $(j, k)$ (i.e., one of $r_j(k + 1), \ldots, r_j(k + 1/\varepsilon)$). For each step $t$, the number of thresholds crossed during that step is equal to the number of threshold collections $(j, k)$ which hit step $t$. (In particular, the step cannot cross more than one threshold in a single cup, and thus cannot cross more than two thresholds from the same collection.)

Let $a_j(i)$ denote the total amount of water poured into cup $j$ in the first $i$ steps. We say that a threshold collection $(j, k)$ has \defn{hitting capacity} $s$ if $\min(a_i, k + 1/\delta + 1) - \max(a_{i - l + 1}, k + 1) = s$. In other words, during steps $i - l + 1, \ldots, i$, a total of $s$ units of the poured water have the possibility of crossing one of the thresholds in the collection $(j, k)$. Notably, the sum of the hitting capacities of all threshold collections is at most $(1 - \varepsilon)pl$, upper bounded by the total amount of water poured during the steps. Moreover, if a threshold collection $(j, k)$ has hitting capacity $s < 1$, then at most one of the steps can be hit by the collection $(j, k)$, and the probability of any step being hit by the collection $(j, k)$ is $s$ (with the outcome depending on $s_j(k)$).

Using this notion of hitting capacity, we now define a partition $Q_1, \ldots, Q_{pl + 1}$ of all threshold collections $(j, k)$. If a threshold collection $(j, k)$ has hitting capacity $s \ge 1$, then we assign it to its own component of the partition. These components of the partition are called \defn{single-threshold components}. The threshold collections $(j, k)$ with hitting capacity less than one are placed in components such that the sum of the hitting capacities within any component is between 1 and 2 (although one of the components may have hitting capacity less than one as an edge case). These components of the partition are called the \defn{multi-threshold components}. Note that the resulting partition has at most $pl + 1$ components since the sum of the hitting capacities in each component is at least $1$, with the exception of at most one component. Moreover, we may assume that the partition has exactly $pl + 1$ components by adding empty components as necessary.

Define the random variables $X_1, \ldots, X_{pl + 1}$ so that $X_t$ is the set of pairs $((j, k), r)$ of threshold collections $(j, k) \in Q_t$ and steps $r \in \{i - l + 1, \ldots, i\}$ such that step $r$ is hit by threshold collection $(j, k)$. Notice that the $X_t$'s are independent of one-another, since they depend on disjoint sets of threshold collections. Moreover, if $T$ is the sum of $T_i + T_{i - 1} + \cdots + T_{i - l + 1}$, then $T$ can be expressed as a deterministic function of the $X_t$'s, with
$$T = F(X_1, \ldots, X_{pl + 1}).$$

Additionally, for a given value $A$ of $X_t$,
\begin{equation}
|F(X_1, \ldots, X_{t - 1}, A, X_{t + 1}, \ldots , X_{pl + 1}) - F(X_1, \ldots, X_{t - 1}, \emptyset, X_{t + 1}, \ldots , X_{pl + 1})| \le |A|,
\eqlabel{eqalmostmcdiarmid}
\end{equation}
since adding any element $((j, k), m)$ to $X_t$ can affect at most one of $T_i, T_{i - 1}, \ldots, T_{i - l + 1}$. Namely, it can only affect $T_m$, and will either increase $T_m$ by one, or not affect $T_m$ at all.

If each $|X_t|$ were guaranteed to be small, then \eqref{eqalmostmcdiarmid} would allow us to apply McDiarmid's inequality to $F(X_1, \ldots, X_{pl + 1})$. As a step in this direction, the following lemma gives a probabilistic bound on $|X_t|$.

\begin{lem}
For any component $Q_t$ and any $r \ge 0$, with high probability in $r$,
$$|X_t| \le O(\log r + 1/\delta).$$
\label{lemboundXt}
\end{lem}
\begin{proof}
For single-threshold components $Q_t$, we are guaranteed that $|X_t| \le 1/\delta$, since a single threshold collection can hit at most $1/\delta$ steps.

Now suppose $Q_t$ is a multi-threshold component. For each threshold collection $(j, k) \in Q_t$, let $Y_{j, k}$ be the random variable counting the number of steps in $\{i - r + 1, \ldots, i\}$ that $(j, k)$ hits. Recall that because $Q_t$ is a multi-threshold component, the threshold collection $(j, k)$ must have some hitting capacity $s < 1$, and will thus either hit exactly one of the steps with probability $s$, or will hit none of the steps with probability $(1 - s)$ (with the outcome depending on the value of the random variable $s_j(k)$). Thus each $Y_{j, k}$ is an indicator variable which takes value $1$ with probability equal to the hitting capacity of $(j, k)$. Moreover, the random variables $Y_{j, k}$ are independent, since the random variables $s_j(k)$ are independent. Thus the size of $X_t$ can be expressed as the sum
\begin{equation}
|X_t| = \sum_{(j, k) \in Q_t} Y_{j, k}
\eqlabel{eqsizeXt}
\end{equation}
of independent indicator variables. Moreover, $\E[X_t]$ will be the sum of the hitting capacities of the threshold collections $(j, k) \in Q_t$, which is at most $2$. In order to apply a Chernoff bound to the summation \eqref{eqsizeXt}, we must also note that there are only finitely many $(j, k)$ with hitting capacity $s > 0$, so \eqref{eqsizeXt} can be regarded as a finite sum. Applying a Chernoff bound, we see that for any $c > 1$, 
$$\Pr[|X_t| > 2 + c \log r] \le e^{-c\log r / 3},$$
completing the proof.
\end{proof}

Let $r$ be a quantity that we will determine a value for later. Define
random variables $X_1', \ldots, X_{p + 1}'$ such that each $X_t'$ is
selected independently from the same distribution as $X_t$, except
restricted to the case where $|X_t| \le c (\log r + 1/\delta)$, for
some sufficiently large constant $c$. In particular, we select $c$ to
be sufficiently large so that Lemma \ref{lemboundXt} guarantees that
each $X_t$ satisfies $|X_t| \le c(\log r + 1/\delta)$ with probability
at least $1 - 1/r$.

Consider the random variable
  $$A = F(X_1', \ldots, X_{pl + 1}').$$
  By \eqref{eqalmostmcdiarmid}, we may apply McDiarmid's inequality to get
  \begin{equation}
  \begin{split}
    \Pr[A \ge \E[A] + R] & \le \exp\left(-\Omega\left(\frac{R^2}{(\log^2 r + 1/\delta^2) \cdot lp}\right)\right) \\
    & \le \exp\left(-\Omega\left(\frac{\delta^2 R^2}{\log^2 r \cdot lp}\right)\right).
  \end{split}
  \eqlabel{eqboundA2}
  \end{equation}
  Recall that each $X_t$ satisfies $|X_t| \le c(\log r + 1/\delta)$
  with probability at least $1 - 1/r$. Since there are $pl + 1$
  such $X_t$, it follows by the union bound that with probability at
  least $1 - O(pl / r)$, $|X_t| \le c (\log r + \delta^{-1})$ for all
  $t$.  If we define the random variable $B = T_{i - l + 1} + \cdots +
  T_i = F(X_1, \ldots, X_{pl + 1})$, then it follows from
  \eqref{eqboundA2} that
  \begin{equation*}
    \begin{split}
      \Pr[B \ge \E[A] + R] & \le \exp\left(-\Omega\left(\frac{\delta^2 R^2}{\log^2 r \cdot lp}\right)\right) + O\left(\frac{pl}{r}\right) \\
       & \le pl \cdot \left(\exp\left(-\Omega\left(\frac{\delta^2 R^2}{\log^2 r \cdot lp}\right)\right) + O(e^{-\ln r})\right). \\
    \end{split}
  \end{equation*}
  Plugging in $\ln r = \left(\frac{\delta^2 R^2}{lp}\right)^{1/3}$, it follows that
  \begin{equation}
    \Pr[B \ge \E[A] + R] \le O(pl) \cdot \exp\left(-\Omega\left(\left(\frac{\delta^2 R^2}{lp}\right)^{1/3}\right)\right).
    \eqlabel{eqprelemea}
  \end{equation}

  Note that \eqref{eqprelemea} is only useful when $pl / r$ is
  small. When this is the case, we can bound $\E[A]$ to be no more
  than $\delta r / 8$, thereby turning \eqref{eqprelemea} into a
  statement about $B$ only.
  \begin{lem}
    Suppose that $pl / r < 1/2$. Then,
    $$\E[A] \le \delta l / 8.$$
    \label{lemea2}
  \end{lem}
  \begin{proof}
    Defining $B$ as above, we begin by considering $\E[B] = \E[T_i + T_{i - 1} + \cdots + T_{i - l + 1}]$. By Lemma \ref{lemexpTi}, and by linearity of expectation,
    $$\E[T_i + T_{i - 1} + \cdots + T_{i - l + 1}] \le O(l \cdot p
 \cdot   e^{-\varepsilon^2 p / 6}).$$ Since $\varepsilon \le p^{-1/3}$, this is
$l \cdot O(e^{-\Omega(\varepsilon^2 p)})$.  Using the fact that
    $\delta$ is sufficiently large in $\Omega(e^{-O(\varepsilon^2 p)})$, it
    follows that
    \begin{equation}
\E[B] \le \delta l / 16.
\eqlabel{eqEB}
    \end{equation}
    On the other hand, $\E[B] = \E[F(X_1, \ldots, X_{pl + 1})]$ can also
    be expressed as,
    \begin{align*}
      &  \E[F(X_1', \ldots, X_{pl + 1}')] \cdot \Pr[X_1, \ldots, X_{lp + 1} \le c(\log r + \delta^{-1})] \\
      & + 
      \E[F(X_1, \ldots, X_{pl + 1}) \cdot \mathbb{I}(X_t > c (\log r + \delta^{-1}) \text{ for some }t)] \\
         & = \E[A] \cdot \Pr[X_1, \ldots, X_{lp + 1} \le c(\log l + \delta^{-1})] + 
      \E[F(X_1, \ldots, X_{pl + 1}) \cdot \mathbb{I}(X_t > c (\log l + \delta^{-1}) \text{ for some }t)] \\
      & \ge  (1 - pl / r) \cdot \E[A] \\
      & \ge \frac{1}{2} \E[A],
    \end{align*}
    where the final step uses that $pl / r \le 1/2$. Applying \eqref{eqEB}, it
    follows that
    $\E[A] \ge \delta l/8,$ as desired.
    
  \end{proof}
  
  Applying Lemma \ref{lemea2} to \eqref{eqprelemea}, we get that when
  $pl/r < 1/2$,
  \begin{equation}
    \Pr[B \ge \delta l / 8 + R] \le O(pl) \cdot \exp\left(-\Omega\left(\left(\frac{\delta^2 R^2}{lp}\right)^{1/3}\right)\right).
    \eqlabel{eqlemeaconsequence}
  \end{equation}
  On the other hand, when $pl/r = pl \cdot \exp\left(-\left(\frac{\delta^2
    R^2}{lp}\right)^{1/3}\right) \ge 1/2$, \eqref{eqlemeaconsequence}
  is trivially true since the left side is upper-bounded by $1$. Thus the equation
  holds for any value of $pl / r$.

  Recall that $l$ is a height-$t$ backlog witness for step $i$ if $2B
  = 2\sum_{m = i - l + 1}^i T_m \ge \delta l + t$. Thus, the
  probability that $l$ is a height-$t$ backlog witness for step $i$ is
  at most
  $$\Pr[B \ge \delta l / 8 + \left(\delta l / 8 + t / 2\right)] \le O(pl) \cdot \exp\left(-\Omega\left(\left(\frac{\delta^2 \left(\delta l + t\right)^2}{lp}\right)^{1/3}\right)\right).$$

  This completes the proof of Proposition \ref{propwitnessprob}.

  \begin{rem}
    In later sections it will be useful note that we have
    proven a statement slightly stronger than Proposition
    \ref{propwitnessprob}. Namely, that
    $$\Pr\left[\sum_{m = i - l + 1}^i T_m > \delta l / 4 + t / 2\right]
    \le O(pl) \cdot \exp\left(-\Omega\left(\left(\frac{\delta^2
      \left(\delta l +
      t\right)^2}{lp}\right)^{1/3}\right)\right).$$
    \label{rempropwitnessdetail}
  \end{rem}

%% \todo{Continue reformatting from here}
  
%% In Section \ref{multismoothefc}, we will extend Theorem
%% \ref{thmmulti-core} to consider the probability of an arbitrary
%% backlog $k$ occurring. In this setting it will be useful to have a
%% general bound on the probability of a height-$t$ backlog witness existing. We conclude the section
%% with such a bound:

  \section{Bounds for Non-Constant Backlogs}\seclabel{multicoresmoothefc}

  Recall that the algorithm given in
  \secref{multicoreconstantbacklog} allows for the emptier to select
  cups in an arbitrary fashion, as long as preference is given to cups
  $j$ satisfying $w_j \ge 1 + \delta$ (over all other cups) and to
  cups satisfying $w_j \neq 0$ (over cups satisfying $w_j = 0$). If
  within each priority level, the emptier always selects the fullest
  cups possible to empty out of, then we call the resulting algorithm
  the \textbf{\multicoresmoothefc}.

  In this section, we prove a probabilistic upper bound for the
  backlog using \multicoresmoothefc. In 
  \secref{lowerbounds}, we will see that when $\epsilon$ is constant,
  our bound is tight.

  As in \secref{multicoreconstantbacklog}, we continue to use
  the \tweakedemphmulticore. The goal of the section is to prove the following theorem:

  \begin{thm}
    Suppose that $1/2 \ge \varepsilon \ge p^{-1/3}$, $\delta \ge
    \frac{1}{\poly(p)}$, and that $\delta < 1$ is sufficiently large in
    $\Omega(e^{-O(\varepsilon^2p)})$. Consider the \multicoresmoothefc on $n$ cups for the \tweakedmulticore. Then for any value $k \ge \log p$ and for
    any step $i$, the backlog is $O(k / \epsilon)$ after step $i$ with probability at least
    $$1 - O\left(e^{-e^{k}}\right).$$
    \label{thmobliviousmulticorefull}
  \end{thm}

  \begin{rem}
    Although Theorem \ref{thmobliviousmulticorefull} does not explicitly
    consider values of $k < \log p$, these values are already implicitly
    covered by Theorem \ref{thmmulti-core}. Indeed, assuming $1/2 \ge
    \epsilon \ge p^{-1/3}$ and $\delta < 1$ is sufficiently large in $e^{-O(\epsilon^2 p)}$, Theorem \ref{thmmulti-core} bounds the probability of
    superconstant backlog by
    $$O\left(e^{-\Omega(\epsilon^2 p)}\right) = O\left(e^{-\Omega(p^{1/3})}\right) \le
    O\left(e^{-e^{\Omega(\log p)}}\right).$$ %% As a trivial consequence, for any value $k
    %% \le \log p$ and for any step $i$, the probability that the maximum
    %% backlog is $O(k / \epsilon)$ after step $i$ is at least
    %% $$1 - O\left(e^{-e^{\Omega(k)}}\right).$$
    \label{remthmoblivsmallk}
  \end{rem}
  
The proof of Theorem \ref{thmobliviousmulticorefull} considers a
dynamic version of the \multicore being played on
the cups containing $4 + \delta$ or more units of water. Since each
such cup $j$ must have a nonzero counter $w_j(i)$, the number of cups
involved in the game at the end of a step $i$ is upper bounded by
$\sum_j w_j(i)$. Building on the results from
\secref{multicoreconstantbacklog}, we obtain a concentration
inequality on the sum $\sum_j w_j(i)$, which then allows us to
complete the proof of Theorem \ref{thmobliviousmulticorefull}.

In Subsection \ref{subsecdynamicmulticore} we present and analyze the
dynamic \multicore. Then in Subsection
\ref{subsecobliviousfull} we combine this with techniques from
\secref{multicoreconstantbacklog} in order to prove Theorem
\ref{thmobliviousmulticorefull}.

\subsection{The Dynamic \MultiCore}\label{subsecdynamicmulticore}

We define the \textbf{dynamic \multicore} as
follows. As in the \untweakedmulticore, at
each step the filler is allowed to distribute up to $(1 - \epsilon)p$
units of water among the cups, placing up to $1 - \delta$ units in any
individual cup; the emptier is then allowed to remove up to one unit
of water from up to $p$ distinct cups. In the \textbf{dynamic} version
of the game, we have the additional caveat that the number of cups
changes dynamically: at the beginning of any turn the filler may
introduce arbitrarily many new cups (into which they must pour a
non-zero amount of water); and at the end of each turn, any empty cups
are removed. Theorem \ref{thmdynamicmulticoredeterministic} provides a
bound for the maximum backlog in terms of the number of active cups
during any given step.

Note that for convenience, we are considering in this subsection the
\untweakedemphmulticore. When applying Theorem \ref{thmdynamicmulticoredeterministic}, we
will therefore have to perform the (simple) task of transferring
between the variants on parameters.

\begin{thm}
Consider an instance of the dynamic \multicore. Let
$r$ be a sufficiently large constant. Let $n_i$ be the number of cups
at the end of step $i$, or $p^r$ if the number of cups is less than
$p^r$. Suppose that $0 < \epsilon < 1$ and that $\delta$ is
sufficiently large in $\Omega(\epsilon / p^{r - 1})$.

Then at the end of each step $i$, the \efc will
achieve maximum backlog $O(\frac{1}{\epsilon} \log n_i)$.
\label{thmdynamicmulticoredeterministic}
\end{thm}
\begin{proof}
  Define $f_i(j)$ to the fill in cup $j$ after step $i$, and define the potential function,
  $$\phi(i) = \sum_{j = 1}^{n_i} \int_{x = 0}^{f_j(i)} (1 + \epsilon)^{\lceil x \rceil} dx.$$

  We will prove that $\phi(i) \le n_i^2$ for all $i$. This, in turn,
  bounds each $f_i(j)$ to be at most $O(\log_{1 + \epsilon}
  \poly(n_i)) \le O(\frac{1}{\epsilon} \log n_i)$.

  Assume by induction that $\phi(i - 1) \le n_{i - 1}^2$ for the $i -
  1$-th step (note that the base case of $i = 1$ is immediate). We
  will use this to prove that $\phi(i) \le n_i^2$.
  
  For a given step $i$, we consider two cases separately:

  \noindent \textbf{Case 1: The emptier removes a full unit of water from each of $p$ cups during step $i$. } In this case, the number of
  cups $n_i$ at the end of step $i$ satisfies $n_i \ge n_{i -
    1}$. Thus it suffices to show that $\phi(i) \le \phi(i - 1)$. This
  follows from Case 1 of the proof of Theorem
  \ref{thmmulticoredeterministic}.
  
  \noindent \textbf{Case 2: At no point during step $i$ do $p$ or more
    cups contain one or more units of water.}  Call the cups that
  contain more than one unit of water at the end of the step the
  \textbf{heavy hitters}. Any water placed in non-heavy hitters during
  step $i$ can increase their contribution to the potential, which is
  initially at most $O(n_i)$, by at most $O(p)$. On the other hand,
  any water placed in heavy hitters is then immediately removed by the
  emptier. Additionally, the emptier removes from each heavy
  hitter at least $\delta$ units of water that were present at the
  beginning of the step.

  Since
  $$\int_{x = 0}^f (1 + \epsilon)^{\lceil x \rceil} dx =
  O\left((1 + \epsilon)^f / \log (1 + \epsilon)\right) =
  O\left(\epsilon^{-1} \cdot (1 + \epsilon)^f\right),$$ this
  removal of $\delta$ units from each heavy hitter will reduce the
  contribution to $\phi(i)$ of each of the heavy hitters by a
  multiplicative factor of $1 - \Omega(\epsilon \delta)$.
  
  Thus there is some $M \in \Omega(\epsilon \delta)$ such that
  $$\phi(i) \le \phi(i - 1) \cdot (1 - M) + M \cdot O(n_i) + O(p),$$
  where the second term is to account for the fact that the
  contributions to $\phi(i)$ by non-heavy hitters are not necessarily
  decreased by the multiplicative factor of $1 - M$, and the third
  term accounts for water added to non-heavy hitters during step $i$.

  Note that $n_i \ge n_{i - 1} - p$. If we recall the inductive
  hypothesis that $\phi(i - 1) \le n_{i - 1}^2$, then
  $$\phi(i) \le (n_i + p)^{2} \cdot (1 - M) + M \cdot O(n_i) + O(p).$$
  Since $n_i \ge p^r$, it follows that 
  \begin{align*}
  \phi(i) & \le (1 + 1/p^{r - 1})^{2} n_i^{2} \cdot (1 - M) + M \cdot O(n_i) + O(p) \\
  & \le (1 + 4/p^{r - 1}) \cdot n_i^{2} \cdot (1 - M) + M \cdot O(n_i) + O(p).
  \end{align*}
  Now using the assumption that $\delta$ is sufficiently large in
  $\Omega(\epsilon^{-1} / p^{r - 1})$ and that $M = \Omega(\epsilon
  \delta)$, we get that $(1 + 4/p^{r - 1}) \cdot (1 - M) \le (1 - M/2)$. Thus

  $$\phi(i) \le n_i^{2} \cdot (1 - M / 2) + M \cdot O(n_i) + O(p).$$
  Using the fact that $n_i \ge p^r$ is sufficiently large in $\Omega(1)$, this gives
  $$\phi(i) \le n_i^{2} \cdot (1 - M / 2) + O(p) \le n_i^{2} \cdot (1 - \Omega(\epsilon \delta)) + O(p).$$
  Since $n_i \ge p^r$, the fact that $\delta$ is sufficiently large in
  $\Omega(\epsilon^{-1} / p^{r - 1})$ then implies that
  $$\phi(i) \le n_i^2,$$
  completing the proof.
\end{proof}

\subsection{Proof of Theorem \ref{thmobliviousmulticorefull}}\label{subsecobliviousfull}

  Define $K = e^k$. We claim that, after a given step $i$, if $\sum_j
  w_j \le \poly(K)$, then the backlog is at most $O(k)$. In
  particular, consider at each step only the cups containing $4 +
  \delta$ or more units of water. By Lemma \ref{lemcounterbacklog},
  each such cup $j$ has counter $w_j \ge 1 + \delta$, meaning that
  the \multicoresmoothefc will always empty out of the fullest of the cups
  containing $4 + \delta$ or more units of water. If we consider just
  these cups, and subtract out $4 + \delta$ from each of their fills,
  then we are playing a dynamic \multicore, and
  using the \multicoresmoothefc. By Theorem
  \ref{thmdynamicmulticoredeterministic}, if there are $r$ cups
  containing $4 + \delta$ or more units of water, then the backlog is at most
  $$O\left(\frac{1}{\epsilon} \max(\log p, \log r)\right).$$ (Note
  that we are actually applying the theorem to $\epsilon' \ge \epsilon$,
  $\delta' = \frac{2\delta}{1 + 2\delta} = \Theta(\delta)$ and to $p'
  = p + 1$; moreover, the theorem requires that $\delta' \ge
  \epsilon^{-1} / \poly(p')$, which is true here because $\epsilon \ge
  p^{-1/3}$ and $\delta \ge 1/\poly(p)$ by assumption.) Since $\ln K =
  k \ge \log p$, it follows that if there are $\poly(K)$ cups in the
  dynamic game, then the backlog is at most $O(\frac{1}{\epsilon}\log
  K) = O\left(\frac{1}{\epsilon}k\right)$. This, in turn, implies that
  if $\sum_j w_j(i) \le \poly(K)$, then the backlog is at most
  $O(\frac{1}{\epsilon}k)$.

  To complete the proof, it suffices to show that $\Pr[\sum_j w_j(i) >
    K^c] \le O\left(e^{-e^k}\right)$ for some sufficiently large
  constant $c$. By Lemma \ref{lemboundwsum}, the probability that
  $\sum_j w_j(i) \ge K^c$ is at most the probability that there exists
  a height-$K^c$ backlog witness for the step $i$. The following
  lemma, which we will prove shortly, can be used to bound the probability of such a witness
  existing.

  \begin{lem}
  Suppose $1/2 \ge \varepsilon \ge p^{-1/3}$ and $\delta < 1$ is
  sufficiently large in $\Omega(e^{-O(\varepsilon^2p)})$. Consider $t
  \ge \frac{\sqrt{p}}{\delta}$. Then the probability that any $l > 0$ is a
  height-$t^4$ backlog witness for step $i$ is at most
  $$O\left(\exp \left(-\Omega\left(-t^{1/3}\right)\right)\right).$$
  \label{lemboundwitness}
\end{lem}

  By Lemma \ref{lemboundwitness}, for $c$ a large enough
  constant, the probability that a height-$K^c$ backlog witness exists
  is at most $O(e^{-K}) = O(e^{-e^{k}})$, as desired. (Note that
  Lemma \ref{lemboundwitness} requires that $t = K^{c / 4}$
  satisfies $t \ge \sqrt{p} / \delta$; since $K = e^k \ge p$ and
  $\delta \ge \frac{1}{\poly(p)}$, this will be true for $c$ a
  sufficiently large constant.)

  In order to complete the proof of Theorem
  \ref{thmobliviousmulticorefull}, it therefore suffices to prove
  Lemma \ref{lemboundwitness}.
\begin{proof}[Proof of Lemma \ref{lemboundwitness}]
  By the union bound and Proposition \ref{propwitnessprob}, the probability that any $0 < l \le t^3$ is a height-$t^4$ backlog witness for step $i$ is at most
  \begin{align*}
    & \sum_{l = 1}^{t^3} O(pl) \cdot \exp\left(-\Omega\left(\left(\frac{\delta^2 t^8}{lp}\right)^{1/3}\right)\right) \\
    & \le \sum_{l = 1}^{t^3} O(pl) \cdot \exp\left(-\Omega\left(\left(\frac{t^6}{l}\right)^{1/3}\right)\right). \\
    & = \sum_{l = 1}^{t^3} O(pl) \cdot \exp\left(-\Omega\left(\frac{t^2}{l^{1/3}}\right)\right). \\
    & \le \sum_{l = 1}^{t^3} O(pl) \cdot \exp\left(-\Omega\left(t\right)\right). \\
    & \le O\left(e^{-\Omega(t)}\right),
  \end{align*}
  where the last step uses the fact hat $pl \le pt^3 \le t^5$.

  The probability that any $l \ge t^t$ is a height-$t^4$ backlog witness for step $i$ is at most
  \begin{align*}
    & \sum_{l = t^3}^\infty O(pl) \cdot \exp\left(-\Omega\left(\left(\frac{\delta^2 (\delta l)^2}{lp}\right)^{1/3}\right)\right) \\
    & \le \sum_{l = t^3}^\infty O(l^2) \cdot \exp\left(-\Omega\left(\left(\frac{\delta^4 l}{p}\right)^{1/3}\right)\right) \\
    & \le \sum_{l = t^3}^\infty O(l^2) \cdot \exp\left(-\Omega\left(\left(\frac{\delta^4 t^2 l^{1/3}}{p}\right)^{1/3}\right)\right) \\
    & \le \sum_{l = t^3}^\infty O(l^2) \cdot \exp\left(-\Omega\left(\left(l^{1/3}\right)^{1/3}\right)\right) \\
    & \le \sum_{l = t^3}^\infty O(l^2) \cdot \exp\left(-\Omega\left(l^{1/9}\right)\right) \\
    & \le \sum_{l = t^3}^\infty O\left(\exp\left(-\Omega\left(l^{1/9}\right)\right)\right) \\
    & \le O\left(\exp\left(-\Omega\left((t^3)^{1/9}\right)\right)\right) \le O\left(\exp\left(-\Omega\left(t^{1/3}\right)\right)\right). \\  
  \end{align*}
  Combining the two probabilities, the probability of any height-$t^4$ backlog witness existing is at most $O\left(e^{-\Omega(t^{1/3})}\right)$.
\end{proof}

%!TEX root =  main.tex

\section{Lower Bounds}
\seclabel{lowerbounds}

In this section we discuss strategies that the filler can follow in
order to maximize backlog, regardless of the algorithm followed by the
emptier. When $\epsilon$ is a positive constant, the resulting bounds
establish that the algorithms for the emptier given in the preceding
sections are essentially optimal, at least up to constant-factor
changes in the backlog and in $p$.

Consider a cup-emptying game on $n$ cups in which at each step, the
filler is allowed to distribute $p / 2$ units of water among the
cups, and then the emptier is allowed to select $p$ cups and remove
all of the water from each of them. We call this the \textbf{universal
  emptying game} on $p$ processors and $n$ cups. Lower bounds for
universal emptying game easily port to lower bounds for \singlecore and \multicore (played with arbitrary $\epsilon$ for the
single-processor variant, and with arbitrary $\epsilon \le 1/2$ and
$\delta$ for the multi-processor variant).

We begin with a lower bound for the universal emptying game in the case where
the filler is an adaptive adversary.

\begin{thm}
  Consider the universal emptying game on $p$ processors and $n$ cups,
  where $n$ is a multiple of $p$. There is an adaptive
  strategy for the filler that accomplishes a backlog of $\Theta\left(\log
  \frac{n}{p}\right)$ after step $\frac{n}{p} - 1$.
  \label{thmlowerbound1}
\end{thm}
\begin{proof}
  The strategy that the filler follows is a simple variant on the
  strategy used by past authors in the \singlecore
  \cite{BenderFeKr15, DietzRa91}. At the beginning of each step $i$, the
emptier will have so far removed water from $p \cdot (i - 1)$
cups. The filler ignores these cups, and distributes $p / 2$ units of
water among the $n - p \cdot (i - 1)$ cups that have never had water
removed from them. This continues until the end of the
$\frac{n}{p}$-th step, at which point every cup has been emptied out
of.

  After the $i$-th step, there will be $n - p \cdot i$ cups each
  containing
  $$\frac{p / 2}{n} + \frac{p / 2}{n - p} + \frac{p / 2}{n - 2p} + \cdots +
  \frac{p / 2}{n - p \cdot (i - 1)}$$ water. It follows that at the end of
  step $\frac{n}{p} - 1$ some cup has fill at least
  \begin{align*}
    & \frac{p}{2} \cdot \left( \frac{1}{n} + \frac{1}{n - p} + \cdots + \frac{1}{2p}\right) \\
    & = \sum_{j = 2}^{n / p} \frac{1}{2j} = \Theta\left(\log \frac{n}{p} \right).
  \end{align*}
\end{proof}

Theorem \ref{thmlowerbound1} implies that our backlog upper bound of
$O\left(\frac{1}{\epsilon} \log n\right)$ (given by Theorem
\ref{thmmulticoredeterministic}) for the deterministic \multicore is
optimal (up to constant factors in the backlog) when $\epsilon \le
1/2$ is a positive constant and $n \ge p^2$.  This is captured
formally in the following corollary.

\begin{corollary}
  For $n \ge p^2$, any deterministic emptying strategy for the
  \multicore played on $n$ cups must have
  worst-case backlog at least $\Omega(\log n)$.
\end{corollary}

Next, we generalize Theorem \ref{thmlowerbound1} to the case where the
filler is an oblivious adversary, meaning the filler is unable to see
which cups the emptier has or has not removed water from.
\begin{thm}
  Consider the universal emptying game on $p$ processors and $n$
  cups. For any $k$ such that $2p \le k \le n$ and such that $k$ is a
  multiple of $p$, there is an obvious strategy for the filler that
  accomplishes a backlog of $\Theta\left(\log \frac{k}{p}\right)$
  after step $\frac{k}{p} - 1$ with probability at least
  $$\left(\frac{1}{\binom{k}{p}}\right)^{k / p - 1} \ge \frac{1}{k^k}.$$
  \label{thmlowerbound2}
\end{thm}
\begin{proof}
  In order to achieve a backlog of $\Theta\left(\log
  \frac{k}{p}\right)$, the filler considers only the cups $1, \ldots,
  k$ and attempts to follow the same construction as presented in the
  proof of Theorem \ref{thmlowerbound1}. The caveat is that the filler
  must now guess during each step which $p$ cups the emptier selected
  in the previous step. Since the filler is concerned only about the
  first $k$ cups, it suffices for the filler to guess after each step
  $i$ a $p$-element subset $S_i \subseteq [k]$ such that during the
  $i$-th step, the emptier did not touch any of the cups in $[k]
  \setminus S_i$. The number of options for each $S_i$ is
  $\binom{k}{p}$. Thus the filler can successfully simulate the
  adaptive lower-bound construction with probability at least
  $$\left(\frac{1}{\binom{k}{p}}\right)^{k / p - 1} \ge \left(\frac{1}{k^p}\right)^{k / p - 1} \ge \frac{1}{k^k}.$$
\end{proof}

For $k \ge p^2$, the oblivious filler in Theorem \ref{thmlowerbound2}
achieves a backlog of $\Omega(\log k)$ with probability at least
$\frac{1}{2^k}$. Substituting $k$ with $2^j$, it follows that for $j
\ge 2 \log p$, the oblivious filler achieves a backlog of $\Omega(j)$
with probability at least $2^{-2^j}$. This, in turn, implies that
Theorems \ref{thmloglog} and \ref{thmobliviousmulticorefull}, which
upper bound the backlog of the \smoothefc and the \multicoresmoothefc,
are optimal (up to constant factors in the backlog) when $\epsilon \le 1/2$ is a positive constant.%% \footnote{In
%% fact, although Theorem \ref{thmobliviousmulticorefull} is formulated only
%% for backlog $j \ge \log p$, the bound of $O\left(e^{-e^{\Omega(\log
%%     p)}}\right)$ that it establishes for backlog $O(\log p)$ continues
%% to hold for all smaller backlogs, as shown in Remark
%% \ref{remthmoblivsmallk}. This bound also is optimal (since it matches the optimal bound for when $j = \Theta(\log p)$), though only up to constant factors in $\log p$.}
The lower bound
demonstrating the optimality of Theorems \ref{thmloglog} and
\ref{thmobliviousmulticorefull} is captured formally in the following
corollary.

\begin{corollary}
Consider the $p$-processor \multicore on $n$ cups
with $\epsilon \le 1/2$. For any $k$ satisfying $\log p \le k \le \log
n$, there is an oblivious pouring strategy that guarantees after some
particular step that there is a backlog of at least $\Omega(k)$ with
probability at least $2^{-2^k}$.
\end{corollary}

Finally, for arbitrary constants $c$, Theorem \ref{thmlowerbound2} can
be used by the filler to achieve backlog $c$ with probability at least
$e^{-O(p)}$. In particular, a backlog of $c$ corresponds with a
backlog of $\log \frac{2^c \cdot p}{p}$, which by Theorem
\ref{thmlowerbound2} can be achieved after some step $\lceil \frac{2^{O(c)}
  \cdot p}{p} \rceil - 1$ with probability at least
$$\left(\frac{1}{\binom{2^{O(c)} \cdot p}{p}}\right)^{2^{O(c)}} \ge
\left(\frac{1}{2^{\left(2^{O(c)} \cdot p\right)}}\right)^{2^{O(c)}} =
2^{-\left(2^{O(c)} \cdot p\right)}.$$ Since $c$ is constant, this
probability is just $2^{-O(p)}$. Hence Theorem \ref{thmmulti-core},
which ensures in the \multicore a backlog of three or smaller with
probability at least $1 - O\left(e^{-\Omega(\varepsilon^2p)}\right)$,
is optimal (up to constant factors in $p$) when $\epsilon \le 1/2$ is
a positive constant. This is captured formally in the following
corollary.

\begin{corollary}
  Consider the $p$-processor \multicore with
$\epsilon \le 1/2$. For any constant $c$, there is an oblivious
pouring strategy that guarantees after some particular step that there
is a backlog of at least $c$ with probability at least
$2^{-O(p)}$.
\end{corollary}

%!TEX root =  main.tex

\section{Recovery from Bad Starting States}
% MAB: we should really have the labels the same as the file names. I haven't done this change yet.
\label{sec:recovery}

In this section we revisit both the \singlecore and the
\multicore in the situation in which the
starting-state of the cups is non-empty. In particular, suppose that
$b$ units of water have already been dispersed among the cups
arbitrarily before the game begins. We wish to show that in both games
the system can quickly recover from such a starting state.

We will prove the following two theorems:

\begin{thm}
Consider the \singlecore beginning from a starting state in
which $b$ units of water have already been dispersed arbitrarily among
the cups. Suppose $\varepsilon \le 1/2$ and consider some value $k \ge
0$ such that $3\log \frac{1}{\varepsilon} \le k$. Then for $i >
\frac{b}{\varepsilon}$, the probability that the \smoothefc has backlog
$k$ or greater after step $i$ is at most
$$O\left(\frac{1}{2^{2^{\Omega(k)}}}\right).$$
\label{thmsinglecorecatchup}
\end{thm}

\begin{thm}
Consider the \tweakedmulticore beginning from a starting
state in which $b$ units of water have already been dispersed
arbitrarily among cups. Set each counter $w_j(0)$ initially to be the
number of thresholds crossed by the $b$ units of water in each cup $j$.

  Suppose that $1/2 \ge \varepsilon \ge p^{-1/3}$ and that $\delta <
  1$ is sufficiently large in $\Omega(e^{-O(\varepsilon^2p)})$.

  Then the \multicoresmoothefc guarantees after any
  given step $i$ satisfying $i \ge \delta^{-8} \cdot p^2$ and $i \ge
  2b / \delta$, the following is true:
  \begin{itemize}
  \item The backlog is three or less with probability at least
  $$1 - O\left(e^{-\Omega(\varepsilon^2p)}\right) - O\left(e^{-\Omega(i^{1/6})}\right).$$
  \item If we further assume that $\delta \ge 1/\poly(p)$, then for
    any $k \ge \log p$, the backlog is $O(k/\epsilon)$ with
    probability at least
    $$1 - O\left(e^{-e^{k}}\right) - O\left(e^{-\Omega(i^{1/6})}\right).$$
  \end{itemize}
\label{thmmulti-corecatchup}
\end{thm}

The proof of Theorem \ref{thmsinglecorecatchup} largely reuses the
ideas from the analysis in the standard setting:
\begin{proof}[Proof of Theorem \ref{thmsinglecorecatchup}]
  For a step $i$, we begin by considering the event $E_i$ that for
  each of the first $i$ steps $1, \ldots, i$, the integer fill at the
  end of each step has remained positive. Note that if the integer
  fill at the end of a step is positive, then the emptier must have
  successfully removed at least one full unit of water during that
  step. Thus, if event $E_i$ occurs, then the emptier must
  successfully remove a total of at least $i$ units of water in the
  first $i$ steps. Given that the cups contain $b$ units of water
  initially, and that the filler places at most $(1 - \epsilon)i$
  units of water during the steps, it follows that
  $$i \le b + (1 + \varepsilon)i.$$
  Therefore, the event $E_i$ cannot occur for $i > b / \varepsilon$.

  So far we have shown that, during the first $i$ steps, the integer
  fill must at some point hit zero. Using this, we can analyze the
  backlog in precisely the same way as we do for the \singlecore. In particular, starting at the point where the
  integer fill was zero, we consider an instance of the dynamic
  cup game being played on the cups that contain one or more
  units of water (with their fill in the dynamic game being one less
  than their actual fill). By Lemma \ref{lemdynamiccups}, the backlog
  at any given step is at most logarithmic in the number of cups
  engaged in the dynamic game. This, in turn, is at most the integer
  fill. By the analysis in Lemma \ref{lemboundedthresholds}, with
  probability at least
  $$1 - e^{-\Omega(r^{1/3})},$$ the integer fill at the end of step
  $j$ is no greater than $r$.\footnote{Notice that in order for this
    analysis to work, we require that there is some $t < j$ such that
    the integer fill was zero at step $j - t$; but fortunately we have
    already established this.} Plugging in $r = 2^k$, the backlog
  after step $j$ will be $O(k)$ with probability at least $1 -
  O\left(e^{-e^{\Omega(k)}}\right)$.
\end{proof}

Next we prove Theorem \ref{thmmulti-corecatchup}.
\begin{proof}[Proof of Theorem \ref{thmmulti-corecatchup}]
  Let $T_0$ be the number of thresholds crossed (using the \multicoresmoothefc's definition of thresholds) by the initial $b$ units of
  water. For each $m > 0$, define $T_m$ to be the number of surplus
  thresholds crossed during the $m$-th step.

  Suppose there is some step $l \le i$ such that after step $l$ the
  counters $w_j(l)$ are all zero. Then the proof of Theorems
  \ref{thmmulti-core} and \ref{thmobliviousmulticorefull} apply
  without modification, since both proofs rely only on analysis of how
  the counters $w_j$ have changed since the last time they were all
  zero.
  
  In order to complete the proof, it suffices to show that with
  probability at least
  $$1 - O\left(e^{-\Omega(i^{1/6})}\right),$$ there is some step $l \le i$ at which the counters
    $w_j$ are all zero.

  By the same analysis as in Lemma \ref{lemboundwsum}, for each step
  $m$ such that $\sum_j w_j(m) \ge 0$, the net increase in $\sum_j w_j$
  during the step is at most $T_m \cdot (1 + \delta) - \delta$. Thus if
  \begin{equation}
    T_0 + \sum_{m = 1}^i T_m \cdot (1 + \delta) < \delta i,
    \eqlabel{eqT0sum}
  \end{equation}
  then we will be
  guaranteed that there is some step $l \le i$ such that $\sum_j w_j(l)
  = 0$, as desired.

  We will consider $T_0$ and $\sum_{i = 1}^j T_j$ separately, proving the bounds,
  \begin{equation}
    \Pr\left[\sum_{m = 1}^i T_m \ge \delta i / 4\right] \le O(e^{-\Omega(i^{1/6})}).
    \eqlabel{eqbreakT0out1}
  \end{equation}
  and
  \begin{equation}
    \Pr[T_0 \ge \delta i / 2] = 0,
    \eqlabel{eqbreakT0out2}
  \end{equation}
  which together establish \eqref{eqT0sum}. 

  By Proposition \ref{propwitnessprob} and Remark \ref{rempropwitnessdetail}, 
  \begin{equation*}
      \Pr\left[\sum_{m = 1}^i T_m \ge \delta i / 4\right]  \le O(pi) \cdot \exp\left(-\Omega\left(\left(\frac{\delta^4 i}{p}\right)^{1/3}\right)\right).
  \end{equation*}
  Using the fact that $i \ge \delta^{-8} \cdot p^2$, it follows that
  \begin{align*}
    \Pr\left[\sum_{m = 1}^i T_m \ge \delta i / 4\right]  & \le O(pi) \cdot \exp\left(-\Omega\left(i^{1/6}\right)\right) \\
    & \le O\left(\exp\left(-\Omega\left(i^{1/6}\right)\right)\right),
  \end{align*}
  thereby proving \eqref{eqbreakT0out1}.

  On the other hand, the number of thresholds crossed by the initial
  $b$ units of water (i.e., $T_0$) can be bounded as follows: If a cup
  $j$ received $r$ units of water, then it cannot have crossed more
  than $\lfloor r \rfloor$ thresholds, since the threshold $r_j(0)$
  is null for each cup $j$ (meaning there is no
  threshold to cross between $0$ and $1$). Therefore, $T_0$ is
  deterministically at most $b$, which in turn is less than $\delta i
  / 2$. This implies \eqref{eqbreakT0out2}, completing the proof.
\end{proof}

%!TEX root =  main.tex

%% \section{Conclusion}

%% In the presence of a small amount of speed augmentation, the algorithms presented in this paper achieve optimal bounds on backlog in the \singlecore and the \multicore. Moreover, by further exploiting the augmentation, the algorithms are able to recover from arbitrary bad initial starting states, eventually recovering the same guarantees as if they had begun with all cups initially empty. We conclude with two open questions.
%% \begin{itemize}
%% \item  Our analysis of the deterministic \efc for the \multicore is the first to provide provable guarantees without the use of clairvoyance. It remains an open question to analyze the algorithm with neither clairvoyance nor speed augmentation. We conjecture that in this setting, the optimal backlog may grow asymptotically past $\log n$.
%% \item Furthermore, it is interesting to investigate whether randomized algorithms can offer the same guarantees as the \smoothefc and \multicoresmoothefc without the use of any speed augmentation. We conjecture that no algorithm can achieve backlog $o(\log n)$ in this setting.
%% \end{itemize}

\section{Acknowledgments} This work was supported in part by NSF Grants CCF 805476, CCF 822388, CNS
1408695, CNS 1755615, CCF 1439084, CCF 1725543, CSR 1763680, CCF
1716252, CCF 1617618, CCF 1314547 and CCF 1533644; and by an MIT Akamai Fellowship and a Fannie
\& John Hertz Foundation Fellowship.

%%% Local Variables:
%%% mode: latex
%%% TeX-master: "main.tex"
%%% End:

%!TEX root =  main.tex

\subsection*{Acknowledgments}
\label{sec:acknowledgments}

This research was supported in part by NSF Grants 1314547 and 1533644.

%%% Local Variables:
%%% mode: latex
%%% TeX-master: "main.tex"
%%% End:

\bibliographystyle{abbrv}
\bibliography{./all}

\begin{thebibliography}{10}

\bibitem{AdlerBeFr03}
M.~Adler, P.~Berenbrink, T.~Friedetzky, L.~A. Goldberg, P.~Goldberg, and
  M.~Paterson.
\newblock A proportionate fair scheduling rule with good worst-case
  performance.
\newblock In {\em Proceedings of the Fifteenth Annual ACM Symposium on Parallel
  Algorithms and Architectures (SPAA)}, pages 101--108, 2003.

\bibitem{AlonSp04}
N.~Alon and J.~H. Spencer.
\newblock {\em The probabilistic method}.
\newblock John Wiley \& Sons, 2004.

\bibitem{AmirFaId95}
A.~Amir, M.~Farach, R.~M. Idury, J.~A.~L. Poutr{\'{e}}, and A.~A.
  Sch{\"{a}}ffer.
\newblock Improved dynamic dictionary matching.
\newblock {\em Inf. Comput.}, 119(2):258--282, 1995.

\bibitem{AmirFr14}
A.~Amir, G.~Franceschini, R.~Grossi, T.~Kopelowitz, M.~Lewenstein, and
  N.~Lewenstein.
\newblock Managing unbounded-length keys in comparison-driven data structures
  with applications to online indexing.
\newblock {\em SIAM Journal on Computing}, 43(4):1396--1416, 2014.

\bibitem{AzarLi06}
Y.~Azar and A.~Litichevskey.
\newblock Maximizing throughput in multi-queue switches.
\newblock {\em Algorithmica}, 45(1):69--90, 2006.

\bibitem{Bar-NoyFrLa02}
A.~Bar-Noy, A.~Freund, S.~Landa, and J.~S. Naor.
\newblock Competitive on-line switching policies.
\newblock In {\em Proceedings of the Thirteenth Annual ACM-SIAM Symposium on
  Discrete Algorithms (SODA)}, pages 525--534, 2002.

\bibitem{BarNi02}
A.~Bar-Noy, A.~Nisgav, and B.~Patt-Shamir.
\newblock Nearly optimal perfectly periodic schedules.
\newblock {\em Distributed Computing}, 15(4):207--220, 2002.

\bibitem{BaruahCoPl96}
S.~K. Baruah, N.~K. Cohen, C.~G. Plaxton, and D.~A. Varvel.
\newblock Proportionate progress: A notion of fairness in resource allocation.
\newblock {\em Algorithmica}, 15(6):600--625, Jun 1996.

\bibitem{BaruahGe95}
S.~K. Baruah, J.~E. Gehrke, and C.~G. Plaxton.
\newblock Fast scheduling of periodic tasks on multiple resources.
\newblock In {\em Parallel Processing Symposium, 1995. Proceedings., 9th
  International}, pages 280--288. IEEE, 1995.

\bibitem{BenderCrCo18}
M.~A. Bender, J.~Christensen, A.~Conway, M.~Farach-Colton, R.~Johnson, and
  M.-T. Tsai.
\newblock Optimal ball recycling.
\newblock In {\em Proceedings of the Symposium on Discrete Algorithms (SODA)},
  2018.

\bibitem{BenderFeKr15}
M.~A. Bender, S.~P. Fekete, A.~Kr{\"{o}}ller, V.~Liberatore, J.~S.~B. Mitchell,
  V.~Polishchuk, and J.~Suomela.
\newblock The minimum backlog problem.
\newblock {\em Theor. Comput. Sci.}, 605:51--61, 2015.

\bibitem{BodlaenderHuKu12}
M.~H.~L. Bodlaender, C.~A.~J. Hurkens, V.~J.~J. Kusters, F.~Staals, G.~J.
  Woeginger, and H.~Zantema.
\newblock Cinderella versus the wicked stepmother.
\newblock In {\em IFIP International Conference on Theoretical Computer
  Science}, pages 57--71, 2012.

\bibitem{BodlaenderHuWo11}
M.~H.~L. Bodlaender, C.~A.~J. Hurkens, and G.~J. Woeginger.
\newblock The cinderella game on holes and anti-holes.
\newblock In {\em Proceedings of the 37th International Conference on
  Graph-Theoretic Concepts in Computer Science (WG)}, pages 71--82, 2011.

\bibitem{ChrobakCsIm01}
M.~Chrobak, J.~Csirik, C.~Imreh, J.~Noga, J.~Sgall, and G.~J. Woeginger.
\newblock The buffer minimization problem for multiprocessor scheduling with
  conflicts.
\newblock In {\em Proceedings of the 28th International Colloquium on Automata,
  Languages and Programming (ICALP)}, volume 2076, pages 862--874, 2001.

\bibitem{DamaschkeZh05}
P.~Damaschke and Z.~Zhou.
\newblock On queuing lengths in on-line switching.
\newblock {\em Theoretical computer science}, 339(2-3):333--343, 2005.

\bibitem{DietzSl87}
P.~Dietz and D.~Sleator.
\newblock Two algorithms for maintaining order in a list.
\newblock In {\em Proceedings of the Nineteenth Annual ACM Symposium on Theory
  of Computing (STOC)}, pages 365--372, 1987.

\bibitem{DietzRa91}
P.~F. Dietz and R.~Raman.
\newblock Persistence, amortization and randomization.
\newblock In {\em Proceedings of the Second Annual ACM-SIAM Symposium on
  Discrete Algorithms (SODA)}, pages 78--88, 1991.

\bibitem{FischerGa15}
J.~Fischer and P.~Gawrychowski.
\newblock Alphabet-dependent string searching with wexponential search trees.
\newblock In {\em Annual Symposium on Combinatorial Pattern Matching}, pages
  160--171. Springer, 2015.

\bibitem{FleischerKo04}
R.~Fleischer and H.~Koga.
\newblock Balanced scheduling toward loss-free packet queuing and delay
  fairness.
\newblock {\em Algorithmica}, 38(2):363--376, Feb 2004.

\bibitem{Gail93}
H.~R. Gail, G.~Grover, R.~Gu{\'e}rin, S.~L. Hantler, Z.~Rosberg, and M.~Sidi.
\newblock Buffer size requirements under longest queue first.
\newblock {\em Performance Evaluation}, 18(2):133--140, 1993.

\bibitem{GkasieniecKl17}
L.~Gasieniec, R.~Klasing, C.~Levcopoulos, A.~Lingas, J.~Min, and T.~Radzik.
\newblock Bamboo garden trimming problem (perpetual maintenance of machines
  with different attendance urgency factors).
\newblock In {\em International Conference on Current Trends in Theory and
  Practice of Informatics}, pages 229--240. Springer, 2017.

\bibitem{Goldwasser10}
M.~H. Goldwasser.
\newblock A survey of buffer management policies for packet switches.
\newblock {\em {SIGACT} News}, 41(1):100--128, 2010.

\bibitem{GoodrichPa13}
M.~T. Goodrich and P.~Pszona.
\newblock Streamed graph drawing and the file maintenance problem.
\newblock In {\em International Symposium on Graph Drawing}, pages 256--267.
  Springer, 2013.

\bibitem{Kopelowitz12}
T.~Kopelowitz.
\newblock On-line indexing for general alphabets via predecessor queries on
  subsets of an ordered list.
\newblock In {\em Foundations of Computer Science (FOCS), 2012 IEEE 53rd Annual
  Symposium on}, pages 283--292. IEEE, 2012.

\bibitem{LitmanMo05}
A.~Litman and S.~Moran-Schein.
\newblock On distributed smooth scheduling.
\newblock In {\em Proceedings of the seventeenth annual ACM symposium on
  Parallelism in algorithms and architectures}, pages 76--85. ACM, 2005.

\bibitem{LitmanMo09}
A.~Litman and S.~Moran-Schein.
\newblock Smooth scheduling under variable rates or the analog-digital
  confinement game.
\newblock {\em Theor. Comp. Sys.}, 45(2):325--354, June 2009.

\bibitem{LitmanMo11}
A.~Litman and S.~Moran-Schein.
\newblock On centralized smooth scheduling.
\newblock {\em Algorithmica}, 60(2):464--480, 2011.

\bibitem{Liu69}
C.~L. Liu.
\newblock Scheduling algorithms for multiprocessors in a hard real-time
  environment.
\newblock {\em JPL Space Programs Summary, 1969}, 1969.

\bibitem{McDiarmid89}
C.~McDiarmid.
\newblock On the method of bounded differences.
\newblock {\em Surveys in combinatorics}, 141(1):148--188, 1989.

\bibitem{MoirRa99}
M.~Moir and S.~Ramamurthy.
\newblock Pfair scheduling of fixed and migrating periodic tasks on multiple
  resources.
\newblock In {\em Real-Time Systems Symposium, 1999. Proceedings. The 20th
  IEEE}, pages 294--303. IEEE, 1999.

\bibitem{Mortensen03}
C.~W. Mortensen.
\newblock Fully-dynamic two dimensional orthogonal range and line segment
  intersection reporting in logarithmic time.
\newblock In {\em Proceedings of the fourteenth annual ACM-SIAM symposium on
  Discrete algorithms}, pages 618--627. Society for Industrial and Applied
  Mathematics, 2003.

\bibitem{RosenblumGoTa04}
M.~Rosenblum, M.~X. Goemans, and V.~Tarokh.
\newblock Universal bounds on buffer size for packetizing fluid policies in
  input queued, crossbar switches.
\newblock In {\em INFOCOM 2004. Twenty-third AnnualJoint Conference of the IEEE
  Computer and Communications Societies}, volume~2, pages 1126--1134. IEEE,
  2004.

\bibitem{SleatorTa85}
D.~D. Sleator and R.~E. Tarjan.
\newblock Amortized efficiency of list update and paging rules.
\newblock {\em Communications of the ACM}, 28(2):202--208, 1985.

\end{thebibliography}
\clearpage
\appendix
\section{Porting Results Between Parameter Regimes}\label{apptweaked}

In this section, we describe how to transfer results for the
\tweakedemphmulticore to results for the \untweakedmulticore.

Suppose $0 \le \delta \le O(\epsilon)$ and $\epsilon \le 1/2$. Then
the corresponding $p'$, $\epsilon'$, $\delta'$ in the \multicore are
$$\epsilon' = 1 - \frac{(1 - \epsilon) \cdot p}{(1 +
  2\delta) \cdot (p + 1)} = \Theta(\epsilon + \delta) =
\Theta(\epsilon),$$
since $1 - \epsilon'$ is the ratio of water poured to water emptied;
$$\delta' = 1 - \frac{1}{1 + 2\delta} = \Theta(\delta),$$ since $1 -
\delta'$ is the ratio of the amount the filler is permitted to place
in each cup to the amount the emptier is allowed to remove; and $p' =
p + 1$. Moreover, $1$ unit of backlog in the \tweakedemphmulticore
corresponds with $1 + 2\delta$ units of backlog in the \multicore.

It follows that, whenever using the \tweakedemphmulticore, one can
always add the assumption that $0 \le \delta \le O(\epsilon)$ to
obtain an analogous result for the \multicore for values of
$\epsilon$ and $\delta$ that differ by a constant factor from before;
and in fact the assumption that $0 \le \delta \le O(\epsilon)$ is
without loss of generality in each of our result statements, since no
advantage in the guarantees is obtained by increasing $\delta$ beyond
$O(\epsilon)$. Thus all of our randomized results port over immediately to the \untweakedmulticore (with the modification of  $\epsilon$ and $\delta$ by constant factors).

\section{Bounding the Number of Random Bits}\label{apprandombits}

In this section, we briefly explain how to limit the number of random
bits used by our algorithms.

In all of our results, we may assume without loss of generality that
$\epsilon$ and $\delta$ are rational and use $O(\log n)$ bits in their
denominators. Although the filler may place water in arbitrary real
quantities into cups, we can assume without loss of generality that
the filler always places multiples of $\frac{\epsilon}{n^2}$ for the
\singlecore and $\frac{\min(\epsilon, \delta)}{n^2}$ for the
\multicore. In particular, even if this is not the case, we may
simulate it by always rounding up the quantity placed into each cup;
this increases the total amount of water poured by no more than
$\min(\epsilon, \delta)/n$, a modification that can easily be absorbed
by changing $\epsilon$ and $\delta$ negligibly.

Many of our algorithms assume the ability to select random real
thresholds in the range $[0, 1]$. However, since we may assume without
loss of generality (as described above) that the water in each cup is
always a rational number with a $O(\log n)$-bit denominator $d$
determined by $\epsilon, \delta, n$, we can select thresholds with
precision only $\frac{1}{d}$ without changing the behavior of the
algorithms.\footnote{However, to be careful about boundary conditions
  for when a threshold is crossed, one should avoid rounding
  thresholds to be exactly multiples of $\frac{1}{d}$; to select a
  random threshold in $[0, 1]$, one can instead select a member of
  $\{\frac{1}{2d}, \frac{3}{2d}, \frac{5}{2d}, \ldots, \frac{2d -
    1}{2d}\}$ at random.}  Thus at most $O(\log n)$ random bits
are required for each threshold.

\section{Analysis of the Dynamic \SingleCore}
\label{appdynamic}

Here we present the proof of Lemma \ref{lemdynamiccups}, analyzing the dynamic \singlecore. The analysis is a simple variation on the analysis presented by \cite{AdlerBeFr03} for the \efc.

\begin{proof}[Proof of Lemma \ref{lemdynamiccups}]
At the end of a given step $i$, we denote the number of cups in the system by $n_i$, and we denote the average amount of water in the fullest $j$ cups by $\av_i(j)$. (That is, $\av_i(j)$ is $1/j$ times the sum of the amounts of water in each of the fullest $j$ cups.) 

We will prove by induction on $i$, that at the end of the $i$-th step, for all $j \ge 1$,
\begin{equation}
\av_i(j) \le 1 + \left(\frac{1}{j + 1} + \frac{1}{j + 2} + \cdots + \frac{1}{n_i}\right).
\eqlabel{eqdynamichypothesis}
\end{equation}
For convenience, we will include cases where $j > n_i$ in our inductive hypothesis, even though \eqref{eqdynamichypothesis} is trivially true for these cases. As a base case, note that \eqref{eqdynamichypothesis} is immediate for $i = 0$, since initially there are no cups. Suppose the inductive hypothesis holds for step $i - 1$, and we wish to prove it for step $i$.

Consider some value $\av_i(j)$, and examine how it evolves over the course of step $i$, starting initially equal to $\av_{i - 1}(j)$. At the beginning of step $i$, some number of new cups are added to the system. Because these cups are initially empty, $\av_i(j)$ is unchanged. The filler then distributes at most one unit of water among all the cups, thereby increasing each $\av_i(j)$ by at most $1/j$. There are then three cases to consider:
\begin{enumerate}
\item \emph{The fullest cup contains one or fewer units of water:} In this case, the emptier will remove all of the water in the cup, thereby removing the cup from the system. However, because the fullest cup contained no more than one unit of water prior to the removal, it must be that no cup contains more than one unit of water. Thus \eqref{eqdynamichypothesis} remains true.
\item \emph{The fullest cup contains more than one unit of water, and after the removal of a unit, the cup is still among the $j$ fullest cups: }In this case, the initial addition of water to the cups may have increased $\av_i(j)$ by as much as $1/j$, but the subsequent removal of water will have then decreased $\av_i(j)$ by exactly $1/j$, resulting in $\av_i(j) \le \av_{i - 1}(j)$.  Since $n_i \ge n_{i - 1}$, it follows that \eqref{eqdynamichypothesis} still holds.
\item \emph{The fullest cup contains more than one unit of water, and after the removal of a unit, the cup is no longer among the $j$ fullest cups: }In this case, we turn our focus to $\av_i(j + 1)$. The initial addition of water during the step may have increased $\av_i(j + 1)$ by as much as $\frac{1}{j + 1}$. Then, once water was removed from the fullest cup, the cups which were formerly in ranks $2, 3, \ldots, j + 1$ for fullness, become the new fullest $j$ cups. The average fill in these cups is at most what the average fill was in the fullest $j + 1$ cups prior to the emptying. Thus the final value for $\av_i(j)$ is at most $\av_{i - 1}(j + 1) + \frac{1}{j + 1}$. By the inductive hypothesis, this is at most
$$1 + \left(\frac{1}{j + 2} + \frac{1}{j + 2} + \cdots + \frac{1}{n_i}\right) + \frac{1}{j + 1} \le  1 + \left(\frac{1}{j + 1} + \frac{1}{j + 2} + \cdots + \frac{1}{n_i}\right),$$
as desired.
\end{enumerate}
In each of the three cases, \eqref{eqdynamichypothesis} continues to hold, implying by induction on $i$ that it will hold for all $i$ and $j$. Considering the case of $j = 1$, we see that the amount of water in the fullest cup $\av_i(1)$ after the $i$-th step is no greater than $1 + 1/2 + 1/3 + \cdots + 1/n_i \le O(\log n_i)$.
\end{proof}

\end{document}